\documentclass[conference,onecolumn]{IEEEtran}
\usepackage{moreverb}

\usepackage{graphics}
\usepackage{graphicx}
\usepackage{epsfig}
\usepackage{epstopdf}
\usepackage{footnote}
\usepackage{amsthm, amsmath, amssymb}
\usepackage{algorithm,algorithmic}
\usepackage[utf8]{inputenc}
\usepackage{placeins}
\usepackage{subfigure}

\def\urltilda{\kern -.15em\lower .7ex\hbox{\~{}}\kern .04em}

\newtheorem{prop}{Proposition}
\begin{document}

\title{A Fair Scheduling Model for Centralized Cognitive Radio Networks}

\author{\IEEEauthorblockN{Didem Gözüpek}
\IEEEauthorblockA{Department of Computer Engineering\\
Gebze Institute of Technology, Kocaeli, Turkey\\
Email: didem.gozupek@gyte.edu.tr}
\and
\IEEEauthorblockN{Fatih Alagöz}
\IEEEauthorblockA{Department of Computer Engineering\\
Bogazici University, Istanbul, Turkey\\
Email: fatih.alagoz@boun.edu.tr}
}

\maketitle
\begin{abstract}
We formulate throughput maximizing, max-min fair, weighted max-min fair, and proportionally fair scheduling problems for cognitive radio networks managed by a centralized cognitive base station. We propose a very general scheduling model accomplishing goals such as making frequency, time slot, and data rate allocation to secondary users with possibly multiple antennas, in a heterogenous multi-channel and multi-user scenario. Moreover, our schedulers ensure that reliable communication between the cognitive base station and secondary users are maintained, no collisions occur among secondary users, and primary users in the service area of the cognitive base station are not disturbed. Two distinctive features of our fair schedulers are that they provide joint temporal and throughput fairness, and take throughput values experienced by secondary users in the recent past, referred to as window size, into account and use this information in the current scheduling decision. We also propose a heuristic algorithm for our fair schedulers and demonstrate through simulations that our proposed heuristic yields very close solutions to the values obtained from the optimization softwares. Furthermore, we make extensive simulations to evaluate our schedulers' performance in terms of both total throughput and fairness for varying number of secondary users, frequencies, antennas, and window size.
\end{abstract}

\begin{IEEEkeywords}
Resource allocation, scheduling, max-min fairness, weighted max-min fairness, proportional fairness, MAC, dynamic spectrum access, cognitive radio networks
\end{IEEEkeywords}

\maketitle


\section{Introduction}
\label{sec:Intro}
Although the licensed spectrum is at a premium, wireless technologies continuously grow and hence intensify the demand
for the radio spectrum. Recent studies reveal that a very small portion of the licensed spectrum is actually utilized.
This inefficient spectrum utilization has fostered research studies that focus on new communication paradigms
referred to as dynamic spectrum access (DSA). Cognitive radio (CR) concept is the technology that can actualize the DSA
concept \cite{Mitola}. The idea of enhancing the spectrum utilization by opportunistically using spatio-temporally
unoccupied portions of the spectrum constitutes the very \emph{raison d'\^{e}tre} of cognitive radio \cite{akyildiz2006next}.

Nodes in a cognitive radio network (CRN) can be classified as primary (licensed) users and secondary
(cognitive, unlicensed) users. A primary user (PU) is the licensed owner of a particular spectrum band and has
exclusive rights to access it, whereas a secondary user (SU) has the ability to detect portions of the spectrum
temporarily unused by its PU and use them opportunistically while ensuring that PUs are oblivious to the SUs. In the rest
of this paper, we use the terms \textit{secondary user} and \textit{cognitive user} interchangeably.

An infrastructure-based (centralized) CRN consists of a cognitive base station (CBS) and SUs that are managed by the
CBS. Scheduling methods in traditional infrastructure-based wireless networks (e.g. cellular networks) cannot be
implemented in CRNs due to the unique features and challenges introduced by the DSA concept. Altering channel
availability owing to PU activity and the requirement that PUs have to be oblivious to SUs make the required
scheduling mechanisms in CRNs different from conventional wireless networks. The fact that PUs need to be unaware of the SUs imply that PUs should not be modified. Hence, designing scheduling mechanisms for CRNs in line with this requirement is an open issue that researchers need to address.

In this work, we propose four scheduling
mechanisms for centralized cognitive radio networks, namely throughput maximizing, max-min fair, weighted max-min fair, and
proportionally fair schedulers. Common features of our proposed schedulers are that they guarantee that the ongoing
communication of the PUs in the service area of the CBS is not disrupted, reliable communication between the CBS and
the SUs is maintained, and no collisions occur among SUs. The objective of the throughput maximizing scheduler is
to make frequency, time slot, and data rate allocation to SUs while maximizing the total throughput of all SUs in the service area of the CBS. Maximizing the total throughput may favor some SUs with persistently good channel
conditions and may hence lead to the starvation of some other SUs. Therefore, providing fairness among SUs may be
important in practice. To this end, we propose three fair scheduling mechanisms, namely max-min fair, weighted max min
fair, and proportionally fair schedulers. Max-min fair scheduler aims to maximize the throughput of the SU that has the
minimum throughput among all SUs. In weighted max-min fair scheduler, on the other hand, each SU has a target weight
symbolizing its target ratio of the total throughput, where target weights of all SUs sum up to $1$. The goal of
the weighted max-min fair scheduler is to make the resulting throughput ratio of every SU as close to its target weight
as possible. On the one hand, max-min fair scheduler may lead to poor total throughput values especially when there is
an SU with bad channel conditions. On the other hand, throughput maximizing scheduler behaves opportunistically by not
taking into account the individual throughput values. Proportionally fair scheduler aims to provide a tradeoff between
user satisfaction and system revenue by maximizing the product of the SU throughput values. System revenue implies the total utility (benefit) that the system receives, which corresponds to the total throughput of all SUs in our case. A distinctive feature of
our fair schedulers is that they take the throughput values of the SUs in the recent past into account. This way, the
scheduler can compensate for the temporary throughput losses of some SUs in the subsequent scheduling periods. Our
simulation results suggest that this feature also provides the scheduler with the flexibility to increase the total
cell throughput by sacrificing from weighted max-min fairness through increasing the deviation from the target
weights. Moreover, we propose a greedy heuristic algorithm for (weighted) max-min fair and proportionally fair schedulers.
We demonstrate through simulations that our proposed heuristic yields close performance to the results obtained from the
optimization software CPLEX.

Works in \cite{gozupek2009jcn} and \cite{gozupek2011genetic} rely on the interference temperature (IT) model
proposed by Federal Communications Commission (FCC) in \cite{fcc2003docket}. Because the IT model requires the
measurement of interference temperature at PUs and setting of an upper interference limit on the entire frequency
band, it spurred a lot of debate since its inception and received both positive and negative comments. Most of the
negative comments were due to the complexity of its practical implementation at the physical layer. Finally, FCC
abandoned the IT concept \cite{order2007fcc}. In this work, we distance ourselves from the IT debate and rely on a much
simpler physical layer model. Instead of measuring the interference temperature at all the measurement points (PUs) and
setting an upper limit for each frequency band, the CBS in our model only needs to determine whether PUs are
actively using a particular frequency or not. There is a maximum tolerable interference power for each active PU as
opposed to each frequency band in the IT model. This can be accomplished using conventional physical and MAC layer
spectrum sensing mechanisms in the CRN literature \cite{quan2008optimal, jiang2009optimal}.

The major merit of our proposed scheduling model is that it is a very general model accomplishing many tasks at
the same time. Specifically, all schedulers proposed in this paper achieve the following:
\begin{enumerate}
\item Frequency allocation
\item Time slot allocation
\item Data rate allocation
\item Taking into account possibly multiple antennas for data transmission
\item Multi-user environment
\item Multi-channel environment
\item Considering channel heterogeneity, where not only the availability of the channels (frequencies), but
also the information about ``how much available a particular frequency is for a particular SU in terms of maximum
allowed transmission power and data rate" differ for each SU and channel (frequency) pair
\item Guaranteeing reliable communication of SUs with the CBS
\item Ensuring that PUs are not disturbed by SU transmissions
\item Ensuring that no collisions occur among SUs
\item A temporal notion of fairness together with throughput fairness (in fair schedulers)
\item Taking into account the recently experienced throughput values during the current scheduling decision and thereby providing flexibility to increase throughput by sacrificing from fairness (in fair schedulers)
\end{enumerate}

To the best of our knowledge, none of the previous work in the literature encompasses all of the above features.

The remainder of this paper is organized as follows: Section \ref{sec:Related Work} describes the related work, whereas
Section \ref{sec:ProblemForm} introduces our problem formulations and proposed solutions for throughput maximizing,
max-min fair, weighted max-min fair, and proportionally fair schedulers. Section \ref{sec:SimResults} discusses the
numerical evaluation and simulation results, while Section \ref{sec:Conc} concludes the paper.

\section{Related Work}\label{sec:Related Work}
Opportunistic scheduling exploits the time-varying channel conditions in wireless networks to increase the overall
performance of the system. It has received a lot of attention in the general wireless networking domain
\cite{liu2001opportunistic, viswanath2002opportunistic, liu2003framework}. A scheme designed only to maximize the
overall throughput can be unfairly biased, especially when there are users with persistently bad channel conditions.
Therefore, maintaining some notion of fairness, such as max-min and proportional fairness, is a vital criterion that
opportunistic scheduling algorithms should address.

There are various efforts in the literature that addressed fairness in the general domain of wireless networks \cite{ye2011ijcs, hou2011ijcs, yang2012study}. Authors in \cite{huang2001max} proposed a set of algorithms for the assignment of max-min fair shares to nodes in a
wireless ad hoc network. The work in \cite{liu2003opportunistic} proposed schedulers that provide deterministic and
probabilistic fairness by decoupling throughput optimization and fairness guarantees as two distinct blocks. Authors in \cite{nandagopal2000achieving} proposed a fair framework for ad hoc wireless networks via a contention
resolution algorithm. However, none of these works takes into account the unique features of the CRN concept such as
the requirement to ensure PU protection; therefore, they are unsuitable for implementation in CRNs.

The opportunistic nature of the CRN concept lends itself conveniently to the opportunistic scheduling paradigm. In
CRN context, not only the physical channel conditions such as fading and path loss, but also the PU activity is a
determinant of the channel quality of an SU. This way, an opportunistic scheduler provides the ability to
opportunistically utilize from the time-varying PU activity. Likewise, a purely opportunistic scheduler may cause an
SU that persistently has an active PU in its vicinity starve in terms of throughput. Thus, providing fairness again
comes into play here to compensate for the throughput losses of the SUs that have active PUs in their surroundings. In
addition to time, frequency is also a resource that needs to be shared fairly among the SUs in CRNs. Unlike the works
in \cite{liu2001opportunistic, viswanath2002opportunistic, liu2003framework, huang2001max, liu2003opportunistic, nandagopal2000achieving}, our
focus in this paper is on cognitive radio networks. In particular, we formulate a scheduling model
that utilizes the opportunistic nature of the CR paradigm as well as taking fairness notions into account.

Works about scheduling in CRNs can be broadly categorized into two groups: The ones about overlay spectrum
sharing paradigm and the ones about underlay spectrum sharing paradigm. Overlay spectrum sharing treats the
availability of a frequency band as a binary decision; i.e., either the band is available for SU transmission or not, whereas underlay spectrum sharing is based on the idea of ensuring that the interference experienced by a PU is below a tolerable limit. Our scheduling model in this paper falls into the underlay spectrum sharing category.

The work in \cite{tang2008joint}, for instance, concentrated on the spectrum overlay paradigm by modeling the
interference as a multi-channel contention graph (MCCG). Authors in  \cite{lee2010spectrum} proposed a spectrum decision
framework for centralized CRNs by considering application requirements and current spectrum conditions. Unlike our
work, the work in \cite{lee2008joint} focused on inter-cell spectrum sharing and proposed a joint spectrum and power
allocation framework. Both \cite{lee2010spectrum} and \cite{lee2008joint} are based on the spectrum overlay paradigm
since they handle spectrum availabilities as a binary decision, i.e., a spectrum band is either available or not. Our
work is distinct in principle from all works in the literature about overlay spectrum sharing because we take into
account the maximum allowed interference power at PUs.

Most of the scheduling works about underlay spectrum sharing in the literature focused on rate and power allocation
without emphasizing other criteria such as frequency allocation, possibly having multiple antennas, ensuring reliable
communication, and joint temporal/throughput fairness. The work in \cite{nguyen2011effective} proposed a joint
scheduling and power control scheme for centralized CRNs. Authors in \cite{shi2010maximizing} focused on capacity
maximization in multi-hop CRNs. None of the works in \cite{nguyen2011effective, shi2010maximizing} focused on issues
such as fairness or multiple antennas.

Our proposed scheduler in
this paper has resemblance with OFDM systems. In fact, this resemblance is not surprising because OFDMA is seen as a promising candidate technology
for future CRNs \cite{mahmoud2009ofdm}. For instance, authors in \cite{rodrigues2009adaptive} proposed a dynamically configurable framework that combines maximum rate, max-min fairness, and proportional fairness policies to maximize the number of satisfied users, whereas the work in \cite{zhang2008opportunistic} derived opportunistic scheduling policies with utilitarian fairness for OFDM systems. The work in \cite{wang2012resource}  proposes a resource allocation algorithm for multi-user OFDM-based cognitive radio systems with proportional rate constraints. None of these works consider joint temporal and throughput fairness or the case where users potentially have multiple antennas.

The major difference of our work from all these works is that we provide a much
more general and complete scheduling model that achieves frequency, time slot, and data rate allocation,
while at the same time accommodating multiple antennas for data transmission, a heterogenous multiuser and multi-channel
environment, taking into account numerous physical layer information such as fading, path loss, mobility, and
time-varying channels in addition to ensuring that SUs maintain a reliable communication and PUs are not disturbed
by SU transmissions. Moreover, we also have a temporal notion of fairness (by ensuring that each SU is assigned at
least one time slot) in addition to throughput fairness. Concurrently allowing these two different notions of fairness
does not exist in the previous work in the literature
\cite{tang2008joint, nguyen2011effective, shi2010maximizing, peng2006utilization, dianati2007scheduling}.
Furthermore, our schedulers have a windowing mechanism that can be tailored by the CBS at her own discretion to
increase the total throughput by allowing a little deviation from target throughput ratios of the SUs in our weighted
max-min fair scheduler. To the best of our knowledge, none of the previous work in the literature encompasses all of
these features.

\section{Problem Formulations and Proposed Solutions}\label{sec:ProblemForm}

We focus on a time-slotted centralized CRN cell, where the CBS is responsible for the overall coordination of
the SUs. Figure \ref{fig:802_22_NetworkDiagram} illustrates the considered network architecture. The scheduler resides
at the CBS and decides on how many packets and with which frequency each SU will transmit in each time slot. Centralized architecture is useful for effectively managing the spectrum resources in CRNs. Advantages of an infrastructure-based (centralized) architecture for CRNs have been outlined in \cite{derya2012wcmc}.

Figure \ref{fig:OurSchedulingMechanism} demonstrates our cognitive scheduling method. $SU(f^{1})$ and
$PU_{1}(f^{1})$ show that $SU$ and $PU_{1}$ currently use frequency $f^{1}$, while $PU_{2}(f^{2})$ illustrates that
$PU_{2}$ uses frequency $f^{2}$. That is to say, $PU_{1}$ and $PU_{2}$ will be disturbed if any SU transmits using
frequency $f^{1}$ and $f^{2}$, respectively, and the interference power received by each PU is above its maximum
tolerable interference power. We represent the maximum tolerable interference power of $PU_{j}$ for frequency
$f$ by $P_{IF_{max}}^{fj}$, where $IF$ stands for ``interference". In other words, the SU in Figure \ref{fig:OurSchedulingMechanism} does not disrupt
$PU_2$ because their operation frequency is different, while it is possible for $PU_1$ to be disturbed by the SU
since they currently utilize the same frequency $f^{1}$. The goal in this work is to determine the transmission power
and consequently the data rate of every SU for every frequency and time slot such that the PUs whose communication are
active in that particular frequency are not disrupted.

Let us represent by $u_{it}$ the number of packets transmitted by SU $i$ in time slot $t$, by $x_{it}$ the number of
packets in the buffer of SU $i$ at the beginning of time slot $t$, and by $f_{it}$ the frequency used by SU $i$ in time
slot $t$. We denote the locations of SU $i$ and PU $j$ in time slot $t$ by $L_{it}^{SU}$ and $L_{jt}^{PU}$, respectively, and the
location of the CBS by $L_{CBS}$. Furthermore, we denote the fading coefficient of the channel between SU $i$ and the
CBS in time slot $t$ by $h_{it}^{CBS}$ and the fading coefficient of the channel between SU $i$ and PU $j$ in time slot $t$
by $h_{ijt}$. Consequently, the vector of buffer states for a total number of $N$ cognitive nodes is
$\overline{x_{t}}=[x_{1t},x_{2t},\cdots, x_{Nt}]$, the vector of transmitted packets is
$\overline{u_{t}}=[u_{1t},u_{2t},\cdots, u_{Nt}]$, and the vector of frequencies used by SUs is
$\overline{f_{t}}=[f_{1t},f_{2t},\cdots, f_{Nt}]$. Additionally, the vector of SU locations is
$\overline{L_{t}^{SU}}=[L_{1t}^{SU},L_{2t}^{SU},\cdots, L_{Nt}^{SU}]$, and the vector of PU locations is
$\overline{L_{t}^{PU}}=[L_{1t}^{PU},L_{2t}^{PU},\cdots, L_{Mt}^{PU}]$, where $M$ is the total number of PUs in the coverage area of
the CBS. Moreover, the matrix of fading coefficients for the channels between SUs and PUs is symbolized by
$\mathbf{h_{t}^{SU,PU}}=[h_{ijt}]$, which is an $N\times M$ matrix. Similarly, the vector of fading coefficients for
the channels between the SUs and the CBS is $\mathbf{\overline{h_{t}^{SU,CBS}}}=[h_{1t}^{CBS},h_{2t}^{CBS},\cdots,h_{Nt}^{CBS}]$.

FCC has required CR devices to use geolocation in conjunction with database consultation in \cite{order2008fcc}. Hence, CBS can gather the estimated values of $\overline{L_{t}^{SU}}$, $\overline{L_{t}^{PU}}$, and $L_{CBS}$ by using the geolocation database. Authors in \cite{wang2009effects}, for instance, also point out the positive impact of location
awareness in CRNs. Furthermore, as also pointed out by \cite{wang2010resource, digham2008joint}, SUs can estimate channel gains between SUs and PUs ($\mathbf{h_{t}^{SU,PU}}$ in our case) by
employing sensors near all receiving points and make them available at the central controller. Values for
$\mathbf{\overline{h_{t}^{SU,CBS}}}$ can also be estimated in a similar way. The fact that our scheduling model
is designed for a centralized (infrastructure-based) CRN setting facilitates the implementation for the gathering of such physical layer parameters. As a consequence, the scheduler's mapping is
$\gamma(t):[\overline{x_{t}},\overline{L_{t}^{SU}},\overline{L_{t}^{PU}},L_{CBS},\mathbf{h_{t}^{SU,PU}},\mathbf{\overline{h_{t}^{SU,CBS}}}]\rightarrow[\overline{f_{t}}, \overline{u_{t}}]$.

\subsection{Throughput Maximizing Scheduler}

Our goal is to formulate a scheduling problem that maximizes the average total network throughput, while assuring that the PUs in the service area of the CBS are not disrupted, and
reliable communication between SUs and the CBS is maintained. We consider a frame-based scheduling model where each scheduling decision is made for a scheduling period that consists of a certain number of time slots during which the network conditions are assumed to be fairly stable. Our initial step is to find an expression for $U_{if}$, which denotes the maximum number of packets that can be sent by SU $i$ using frequency $f$ in any time slot during the scheduling period.

Let $\mathcal{F}=\{1,2,\cdots,F\}$ denote the set of $F$ frequencies, $\mathcal{N}=\{1,2,\cdots,N\}$
show the set of $N$ SUs, $P_{IF_{max}}^{fj}$ represent the maximum tolerable interference power of PU $j$ for
frequency $f$, and $P_{r_{jt}}^{f}$ symbolize the power received by PU $j$ through frequency $f$ in time slot $t$.
Moreover, $\Phi_{CBS}^{ft}$ denotes the set of PUs that are actively utilizing frequency $f$ in the coverage area of
the CBS in time slot $t$, and $P_{r_{CBS}}^{if_{it}t}$ is the power received by the CBS owing to the possible transmission
of SU $i$ using frequency $f_{it}$ in time slot $t$, where $f_{it}$ is the frequency used by SU $i$ in time slot $t$. Besides, $S$ is the packet size, $B$ is the bandwidth, $\zeta$ is the sum of interference power from the PUs to SU $i$ and noise power, and $T_{s}$ is the time slot length. Interference power from the PUs to the SUs and to the CBS may depend on many factors such as the spectrum occupancy of the PUs as well as the locations. Even when this interference is modeled through different variables, an expression for the maximum transmission rate $U_{if}$ for each SU $i$ and frequency $f$ can be obtained. Once the $U_{if}$ values are obtained, our ILP formulations for our throughput maximizing, (weighted) max-min fair, and proportionally fair schedulers remain the same. Therefore, for notational simplicity, we represent the interference power from the PUs and the noise power by a single variable $\zeta$. On the other hand, units for parameters $S$, $B$, and $T_{s}$ are bits/packet, bits/second, and seconds/time slot, respectively.

In order to ensure that the PUs are not disturbed by the SU transmissions, we have to guarantee that interference power values received by PUs because of SU transmissions
adhere to the tolerable limits. In other words, we need to impose a restriction on the
interference power received by PUs; i.e., we need to have $P_{r_{jt}}^{f} \le P_{IF_{max}}^{fj}$ $\forall j\in\Phi_{CBS}^{ft}$ and $\forall f\in\mathcal{F}$. The primary goal of SUs is to communicate with the CBS. In doing so,
they may increase the interference received by PUs. Therefore, the maximum tolerable interference requirement
of the PUs translates to a maximum transmission power constraint on the SUs because in order to ensure that the
interference perceived by the PUs is within the tolerable limits, SUs need to adjust their transmission powers while
communicating with the CBS. By considering the channel conditions between SUs and PUs, we can impose a maximum
transmission power constraint on SUs. We can represent the maximum
permissible transmission power for each SU $i$ and frequency $f$ in time slot $t$ by
$P_{xmt}^{ift}$. We use free space path loss and fading in modeling the channels between SUs and PUs, as well
as the ones between the CBS and SUs. Free space propagation is a commonly used model in wireless networks since it efficiently encapsulates the effect of distance on the received signal power as an inverse square law. Other propagation models can also be used in our work since only the maximum possible transmission powers for the SUs and consequently the $U_{if}$ values would change and the scheduling problem formulations as well as the heuristic algorithm would stay exactly the same.

We derive the following expression for $P_{xmt}^{ift}$:
\begin{align}
P_{xmt}^{ift} =& \min_{j\in\Phi_{CBS}^{ft}}\dfrac{P_{IF_{max}}^{fj}}{(\dfrac{\lambda_{f}}{4\pi d_{ijt}}\times |h_{ijt}|)^{2}} \label{initDeriv1}
\end{align}
where $P_{IF_{max}}^{fj}$ represents the maximum tolerable interference power of PU $j$ for frequency $f$, $d_{ijt}$ equals the distance between SU $i$ and PU $j$ in time slot $t$, $\lambda_{f}$ is the wavelength of frequency $f$, and $\Phi_{CBS}^{ft}$ symbolizes the set of PUs that are in the coverage area of the CBS and that are carrying out their communication using frequency $f$ in time slot $t$. If there are no PUs in the coverage area of the CBS using frequency $f$ in time slot $t$, we assume that there is one such PU at the boundary of the coverage region of CBS. This way, we ensure that possible PUs in the neighboring cells are not adversely affected by the transmissions in our cell. Moreover, $h_{ijt}$ denotes the fading coefficient of the channel between SU $i$ and PU $j$ in time slot $t$. In particular,
$(\dfrac{\lambda_{f}}{4\pi d_{ijt}})^{2}$ refers to the path loss of the channel between SU $i$ and PU $j$ for frequency $f$ in time slot $t$ as a result of the free space path loss formula. PU $j$ that has the minimum value of $\dfrac{P_{IF_{max}}^{fj}}{(\dfrac{\lambda_{f}}{4\pi d_{ijt}}\times |h_{ijt}|)^{2}}$ is the one that is the most severely affected by data transmission of SU $i$ using frequency $f$ in time slot $t$. Intuitively, this corresponds to a PU which has low interference tolerance and good channel conditions (such as being close in proximity to the SU). If this PU $j$ is guaranteed not to be affected by this SU transmission, other PUs will also not be disturbed. Therefore, maximum transmission power of SU $i$ using frequency $f$ in time slot $t$ is set to $\min_{j\in\Phi_{CBS}^{ft}}\dfrac{P_{IF_{max}}^{fj}}{(\dfrac{\lambda_{f}}{4\pi d_{ijt}}\times |h_{ijt}|)^{2}}$.

A central entity called ``spectrum broker" or ``spectrum policy server" may coordinate the communication among several CBSs \cite{akyildiz2006next, ileri2005demand}. This way, information such as the distance to the PUs in the neighboring CBS cells can be gathered and used in determining the maximum allowed transmission power of the SUs. This way, the CBS scheduler can also ensure that the PUs in the neighboring CBS cells are also not disturbed. This modification only changes the value of the maximum allowed transmission power values for the SUs and hence the $U_{if}$ values, and do not impact any of the scheduling problem formulations or algorithms proposed in this paper. Therefore, without loss of generality, we focus on not disturbing the PUs in the service area of a single CRN cell in this paper.

Let us assume for simplicity, and yet without loss of generality (w.l.o.g), that $S = B \times T_{s}$. Recall that $P_{xmt}^{ift}$ denotes the maximum permissible transmission
power for SU $i$ and frequency $f$ in time slot $t$. We can convert this maximum transmission power constraint to a maximum
data rate constraint on SUs through Shannon's capacity function for Gaussian channels. This conversion ensures that reliable communication between SU $i$ and the CBS is achieved by having the scheduler to choose the number of packets transmitted in each time slot $t$ using some frequency $f$, $U_{ift}$, equal to the Shannon capacity function for a Gaussian channel \cite{cover2006elements}. Recent advances in coding (Turbo codes and LDPC) make it feasible to achieve performance close to Shannon capacity using codes over finite block lengths. Besides, the uncoded systems also follow the exponential relationship between transmission rate and transmit power \cite{rajan2004delay}. All the optimization problems we formulate in this paper take as input the maximum possible transmission rate of each SU for each frequency. If the relationship between transmission rate and power is not exponential, then the calculation of these maximum possible transmission rates will be different; however, our formulated optimization problems in this paper will remain the same. Therefore, even in cases where the relationship between transmission rate and power is not exponential, all of our optimization problems are still valid.

Therefore, the expressions from \eqref{initDeriv2} to
\eqref{initDeriv4} in the following hold:
\begin{align}
G_{ift}^{CBS} \triangleq &  \dfrac{\lambda_{f}}{4 \pi d_{it}^{CBS} \times h_{it}^{CBS}} \label{initDeriv2}\\
P_{r_{CBS}}^{ift} =& P_{xmt}^{ift} \times |G_{ift}^{CBS}|^{2}  \label{initDeriv3}\\
U_{ift} =& \Big \lfloor \ln(1 + ( P_{xmt}^{ift} \times \dfrac{|G_{ift}^{CBS}|^{2}}{\zeta}))\Big \rfloor  \label{initDeriv4}
\end{align}

where $d_{it}^{CBS}$ is the distance between SU $i$ and the CBS in time slot $t$, and $U_{ift}$ is the maximum number of packets that can be transmitted by SU $i$ using frequency $f$ in time slot $t$. Equation \eqref{initDeriv3} relates the power that would be received by the CBS due to the transmission of SU $i$ using frequency $f$ in time slot $t$ ($P_{r_{CBS}}^{ift}$) if SU $i$ uses its maximum permissible transmission power for frequency $f$ and time slot $t$ ($P_{xmt}^{ift}$). Finally, equation \eqref{initDeriv4} is due to
Shannon's capacity function for Gaussian channels, where $P_{r_{CBS}}^{ift}$ is replaced by \eqref{initDeriv3}. Recall that we have assumed that $S=B\times T_{s}$. The floor operation $\lfloor . \rfloor$ in \eqref{initDeriv4} is necessary since $U_{ift}$ can naturally only take integer values.

Notice that we consider both the interference from SUs to PUs and the interference from PUs to SUs. Recall that the variable $\zeta$ models the interference from PUs to SUs, whereas \eqref{initDeriv1} considers the interference from SUs to PUs and determines the values for $P_{xmt}^{ift}$ such that the interference imposed on the PUs is not above the tolerable limit. The reason is that the goal of the scheduler is to govern the transmission of SUs to the CBS without causing any harmful interference to PUs (without disturbing PUs).

Assume that the network conditions; i.e., PU spectrum occupancies, PU and SU locations, and all the channel
fading coefficients, are small enough not to have any influence on the $U_{ift}$ values for a duration of $T$ time
slots in the considered centralized CRN cell. In other words, assume that the $U_{ift}$ value for SU $i$ and frequency
$f$ remain the same for a duration of $T$ time slots, which is equal to the scheduling period during which the network
conditions are fairly stable. Due to the floor operation  $\lfloor . \rfloor$ in \eqref{initDeriv4}, the scheduling period length $T$ does
not obligate the PU spectrum occupancies as well as the PU and SU locations to remain constant in that time period, but
only requires that the change in their values does not alter $U_{ift}$ (related to SU $i$ and frequency $f$) for a
period of $T$ time slots. The value of $T$, in general, hinges upon the characteristics of the spectral environment. For
instance, a slowly varying spectrum environment like the TV bands used by an IEEE 802.22 network, allows $T$ to have a
fairly large value. Therefore, instead of $U_{ift}$, let us use $U_{if}$ , which represents the maximum number of
packets that can be transmitted by SU $i$ using frequency $f$ in every time slot for a duration of $T$ time slots.

After obtaining the $U_{if}$ values by using the analysis in \eqref{initDeriv1}-\eqref{initDeriv4}, we can formulate our throughput maximizing scheduling problem, which is a binary integer linear programming problem (ILP) that takes the $U_{if}$ values as input variables, as follows:
\begin{align}
\max(\sum_{i=1}^{N}\sum_{f=1}^{F}&\sum_{t=1}^{T}\dfrac{U_{if}X_{ift}}{T})\label{TMSObj}\\
s.&t.\notag\\
\sum_{f=1}^{F}\sum_{t=1}^{T}&X_{ift}\ge 1;\forall i\in\mathcal{N}\label{TMSconstr1}\\
\sum_{i=1}^{N}&X_{ift}\le 1; \forall f\in\mathcal{F}, \forall t\in\mathcal{T}\label{TMSconstr2}\\
\sum_{f=1}^{F}&X_{ift}\le a_{i};\forall i\in\mathcal{N},\forall t\in\mathcal{T}\label{TMSconstr3}
\end{align}

where $\mathcal{T}$ denotes the set of $T$ time slots in a scheduling period; i.e., $\mathcal{T}=\{1,2,\cdots,T\}$. In
addition, $X_{ift}$ is a binary decision variable such that $X_{ift}=1$ if SU $i$ transmits with frequency $f$ in time
slot $t$ and $0$ otherwise, and $a_{i}$ is the number of transceivers (antennas) of SU $i$. In this formulation,
\eqref{TMSconstr1} guarantees that at least one time slot is assigned to every SU and hence provides a temporal notion
of fairness, while constraint \eqref{TMSconstr2} ensures that at most one SU can transmit in a particular time slot and frequency. The reason for having constraint \eqref{TMSconstr2} is twofold. First, having more than one SU transmit their data to the CBS using the same frequency in the same time slot leads to collision at the CBS, which is the receiver side, because the CBS cannot differentiate between the signals sent by different SUs if they are in the same frequency and time slot. Therefore, constraint \eqref{TMSconstr2} is a fundamental necessity to prevent collisions among the SUs. Second, constraint \eqref{TMSconstr2} also ensures that the maximum tolerable interference thresholds of the PUs are met. Consider a situation where two SUs transmit with a certain frequency $f$ in a particular time slot $t$. This implies that two SUs will contribute to the value of $P_{r_{jt}}^{f}$ in $P_{r_{jt}}^{f} \le P_{IF_{max}}^{fj}$. However, when we calculate the maximum transmission power for an SU by considering the interference that can occur at PUs, we consider the interference created by only that SU. Therefore, having more than one SU transmit in the same frequency and time slot may increase the aggregate interference experienced by PUs above the maximum tolerable interference limit, which is $P_{IF_{max}}^{fj}$ . Hence, besides avoiding collisions among SUs, \eqref{TMSconstr2} serves the purpose of guaranteeing that the aggregate interference at the PUs is within the tolerable threshold. Moreover, \eqref{TMSconstr3} represents the fact that an SU $i$ cannot transmit at the same time
using frequencies more than the number of its transceivers (antennas), $a_{i}$, because each transceiver can tune to at
most one frequency at a time. We assume for simplicity
and yet w.l.o.g. that channel conditions of each antenna of a particular SU are the same for a particular frequency.
We make this assumption to isolate us from the possible impacts of different channel conditions for different
antennas and concentrate on the performance impact of the number of antennas. Note here that our formulation does not
mandate the SUs to have multiple antennas; in other words, even when each SU has a single antenna, our formulation in
\eqref{TMSObj}-\eqref{TMSconstr3} is still valid since $a_{i} = 1$ $\forall i\in\mathcal{N}$. Having multiple antennas for data transmission enables the SUs to transmit simultaneously using different frequencies.

After the scheduling decisions about $U_{if}$ and $X_{ift}$ values are made, an SU $i$ for which $X_{ift} = 1$ transmits
$\min(x_{it},U_{if})$ number of packets using frequency $f$ in time slot $t$. In the simulations part of this work, we consider traffic in which all flows
are continuously backlogged so that the resulting throughput is completely related to the scheduling process and
channel conditions without any variation because of the traffic fluctuation. That is to say, in the simulations part
of this work, it is always true that $x_{it}\ge U_{if}$, $\forall i\in\mathcal{N}$, $\forall f\in\mathcal{F}$,
$\forall t\in\mathcal{T}$; i.e., each SU always has sufficient number of packets waiting in its buffer to be transmitted to the CBS. This situation is necessary in order to effectively evaluate the performance of the scheduling process by avoiding the possible influence of the traffic arrival process.

\subsection{Max-Min Fair Scheduler}
Our goal is to formulate a scheduling problem that maximizes the throughput of the SU experiencing the minimum
throughput among all SUs, while ensuring that the communication of none of the PUs is disturbed, and reliable
communication between SUs and CBS is achieved. We firstly implement the analysis in
\eqref{initDeriv1}-\eqref{initDeriv4} and obtain the $U_{if}$ values. Secondly, we define $R_{\varphi}^{i}$, the
aggregate average throughput of SU $i$ in the last $\varphi$ scheduling periods. The unit for $R_{\varphi}^{i}$ is
``packets per time slot". All of the  $R_{\varphi}^{i}$ values are initialized to $0$ for all SUs. Let us define
$\overline{R_{\varphi}^{i}}$, which is based on an exponentially weighted low pass filter, as follows:

\begin{align}
\overline{R_{\varphi}^{i}}=
(1-\displaystyle\frac{1}{\min(k,\varphi)})R_{\varphi}^{i} +
\displaystyle\frac{1}{\min(k,\varphi)}\frac{\displaystyle \sum_{f=1}^{F} \sum_{t=1}^{T} U_{if}X_{ift}}{T} \label{varphiUpdate}
\end{align}

Here, $\dfrac{\displaystyle \sum_{f=1}^{F} \sum_{t=1}^{T} U_{if}X_{ift}}{T}$ denotes the throughput of SU $i$ in the
current scheduling period $k$ and $\dfrac{1}{\min(k,\varphi)}$ is the weight given to it. On the other hand,
$(1-\displaystyle\frac{1}{\min(k,\varphi)})$ is the weight given to the value of $R_{\varphi}^{i}$, i.e. the aggregate
average throughput of SU $i$ at the start of the current scheduling period. Using an exponentially weighted low pass
filter enables our scheduler to give more importance to the throughput experienced in the more recent scheduling
periods than the distant past. At the end of each scheduling period $k$, the value of $R_{\varphi}^{i}$ is updated
as $R_{\varphi}^{i}\leftarrow \overline{R_{\varphi}^{i}}$. Since both $k$ and $\varphi$ are constant in a particular
scheduling execution, w.l.o.g we use $\varphi$ instead of $\min(k,\varphi)$ in the rest of this paper.

Our max-min fair scheduling problem is the following mixed ILP, which is executed by the CBS for each scheduling period consisting of $T$ time slots:
\begin{align}
\max & \hspace{1mm}Z \label{mmfsObj}\\
s.&t.\notag \\
Z\le  (1-\dfrac{1}{\varphi})R_{\varphi}^{i} &+ \dfrac{1}{\varphi}\dfrac{\displaystyle\sum_{f=1}^{F}\sum_{t=1}^{T}U_{if}X_{ift}}{T} \label{mmfsConstr1}\\
\eqref{TMSconstr1}, \eqref{TMSconstr2}, &\text{ and } \eqref{TMSconstr3} \label{mmfsConstr2}
\end{align}
where \eqref{mmfsObj} and \eqref{mmfsConstr1} together maximize $\min_{i \in \mathcal{N}} \overline{R_\varphi^i}$.
Keeping track of and using the aggregate throughput information $R_{\varphi}^{i}$ in
\eqref{mmfsObj}-\eqref{mmfsConstr2} in lieu of maximizing the minimum throughput only in that scheduling period enables
us to provide fairness in a longer time scale. If an SU suffers from low throughput due to PU activity in its
vicinity, it can compensate for this loss in the subsequent scheduling periods due to the historical throughput
information accumulated in $R_{\varphi}^{i}$. Embedding the information about the accrued throughput in the scheduling
algorithm itself is a vital feature of our schedulers. The variable $\varphi$ can be regarded as the window size during
which the changes in the network conditions are considered to be important. If $\varphi$ is too large, the scheduler
will be too responsive to small changes in the network conditions. If $\varphi$ is is too small, the scheduler will be
inflexible in compensating for the temporary fluctuations in the network conditions. For instance, having
$\varphi = 1$ connotes that merely the network conditions in the current scheduling period are taken into consideration
without paying any attention to what has happened in the recent past.

\begin{prop}\label{thmMMFS}
Let $\Omega_{MaxMin}^{OPT}$ be the optimal value for the minimum throughput found by the max-min fair scheduler in
\eqref{mmfsObj}-\eqref{mmfsConstr2} for $\varphi=1$. Let $\Omega_{MaxMin}^{UB}$ be the maximum possible value for
$\Omega_{MaxMin}^{OPT}$ for $\varphi=1$, and $\Omega_{ThrMax}^{OPT}$ be the total throughput found by the throughput maximizing
scheduler in \eqref{TMSObj}-\eqref{TMSconstr3} for $\varphi=1$. Then, $\Omega_{MaxMin}^{OPT} \le \Omega_{MaxMin}^{UB} = \dfrac{\Omega_{ThrMax}^{OPT}}{N}$.
\end{prop}
\begin{proof}
Since $\Omega_{MaxMin}^{OPT}$ is the throughput of the SU with minimum throughput, the throughput of any other SU
among the remaining $N-1$ SUs is at least $\Omega_{MaxMin}^{OPT}$. Let $\Omega_{MaxMin}^{tot}$ denote the total
throughput resulting from the max-min fair scheduler execution. Then, $\Omega_{MaxMin}^{tot} \ge \Omega_{MaxMin}^{OPT} + (N-1)\Omega_{MaxMin}^{OPT} = N \times \Omega_{MaxMin}^{OPT}$. Since $\Omega_{MaxMin}^{tot} \le \Omega_{ThrMax}^{OPT}$, it follows that $\Omega_{MaxMin}^{OPT} \le \Omega_{MaxMin}^{UB} = \dfrac{\Omega_{ThrMax}^{OPT}}{N}$.
\end{proof}

\subsection{Weighted Max-Min Fair Scheduler}

Max-min fair scheduling problem in \eqref{mmfsObj}-\eqref{mmfsConstr2} maximizes the average throughput of the SU
that has the minimum aggregate throughput; therefore, all SUs operate at a similar throughput level. However, in some
practical cases, CBS operator may want to differentiate between different SUs and provide them with different levels
of service by giving them priorities. We can quantify these priorities by associating a target weight with each SU such that the higher the target weight of
an SU is, the more it is favored by the CBS scheduler in the frequency and time slot allocation.

Our goal is to formulate a scheduling problem that makes the throughput
ratios of all the SUs as close to their target weights as possible, while ensuring that the communication of none of
the PUs is disturbed, and reliable communication between the SUs and the CBS is achieved. Let us denote by $\eta_{i}$, where $0 \le \eta_{i} \le 1$, the target weight of SU $i$; i.e., the target ratio of the throughput of SU $i$ to the
total throughput of all SUs in the CRN cell. Weights of all SUs in the coverage area of the CBS sum up to $1$; that is
to say, $\displaystyle\sum_{i=1}^{N}\eta_{i} = 1$. If it were theoretically possible to assign all SUs their exact
target weights, then the values obtained via dividing their throughput by $\eta_{i}$, which we refer to here by
\emph{normalized throughput values}, would all be equal to each other. Accordingly, the objective function of our weighted max-min fair scheduling problem aims to maximize the normalized throughput value of the SU that has the minimum normalized
throughput value, and thereby makes the throughput ratio of every SU as close to its target weight as possible.
We firstly implement the analysis in \eqref{initDeriv1}-\eqref{initDeriv4} and obtain the $U_{if}$ values. Secondly, we
calculate $R_{\varphi}^{i}$  $\forall i\in\mathcal{N}$, and update them at the end of every scheduling period as in
\eqref{varphiUpdate}. Our weighted max-min fair scheduling problem is the following mixed ILP, which is executed by the CBS for each scheduling period consisting of $T$ time slots:

\begin{align}
\max\hspace{1mm} &Z' \label{weightedMmfsObj}\\
s.&t.\notag \\
Z' \le & \dfrac{(1-\dfrac{1}{\varphi})R_{\varphi}^{i} + \dfrac{1}{\varphi}\dfrac{\displaystyle\sum_{f=1}^{F}\sum_{t=1}^{T}U_{if}X_{ift}}{T}}{\eta_{i}} \label{weightedMmfsConstr1}\\
\eqref{TMSconstr1},& \eqref{TMSconstr2}, \text{ and } \eqref{TMSconstr3} \label{weightedMmfsConstr2}
\end{align}

where the constraint \eqref{weightedMmfsConstr1} specifies $Z'$, which is the normalized aggregate throughput of the
SU that has the minimum normalized aggregate throughput among all SUs. The objective function in
\eqref{weightedMmfsObj} maximizes this minimum normalized aggregate throughput ($Z'$). In other words,
\eqref{weightedMmfsObj} and \eqref{weightedMmfsConstr1} together maximize
$\min_{i\in\mathcal{N}}\dfrac{\overline{R_{\varphi}^{i}}}{\eta_{i}}$. As in the max-min fair
scheduler, having $\varphi = 1$ implies that only the current scheduling period is taken into consideration. Similarly,
keeping track of and using the aggregate throughput information $R_{\varphi}^{i}$ in
\eqref{weightedMmfsObj}-\eqref{weightedMmfsConstr2} in lieu of maximizing the minimum normalized aggregate throughput
merely in that scheduling period enables us to provide fairness in a longer time scale.

Notice here that in our formulations, max-min fair scheduler turns out to be a special case of the weighted max-min
fair scheduler where $\eta_{i} = \dfrac{1}{N}$ $\forall i\in\mathcal{N}$; i.e., all target weights are equal to
each other. This result intuitively makes sense since the goal of the max-min fair scheduler is to essentially make
the throughput of each SU as close to each other as possible.

\subsection{Proportionally Fair Scheduler}

The notion of proportional fairness aims to provide a tradeoff between users' satisfaction and system revenue
\cite{li2008proportional}. A data rate allocation is said to be proportionally fair if for any other feasible rate
allocation, the aggregate of the proportional changes is not positive. Accordingly, proportional fairness is achieved
by maximizing the sum of the logarithms of the data rates, which is equivalent to maximizing the product of the data
rates \cite{nandagopal2000achieving}.

Our goal is to formulate a scheduling problem that maximizes the products the SU
throughput values, while ensuring that the communication of none of the PUs is disturbed, and reliable communication
between SUs and the CBS is achieved. Firstly, as in the max-min fair and weighted max-min fair
schedulers, we implement the analysis in \eqref{initDeriv1}-\eqref{initDeriv4} and obtain the $U_{if}$ values.
Secondly, we calculate $R_{\varphi}^{i}$ $\forall i \in \mathcal{N}$ and update them at the end of every scheduling
period using \eqref{varphiUpdate}. Our proportionally fair scheduler is the following mixed integer program, which is executed by the CBS for each scheduling period consisting of $T$ time slots:

\begin{align}
\max\hspace{1mm} &(\sum_{i=1}^{N}\log((1-\dfrac{1}{\varphi})R_{\varphi}^{i} + \dfrac{1}{\varphi}\dfrac{\displaystyle\sum_{f=1}^{F}\sum_{t=1}^{T}U_{if}X_{ift}}{T})) \label{pfsObj}\\
s.&t.\notag \\
\eqref{TMSconstr1},& \eqref{TMSconstr2}, \text{ and } \eqref{TMSconstr3} \label{pfsConstr}
\end{align}

where the objective function in \eqref{pfsObj} maximizes the sum of logarithms of the aggregate throughput values of
each SU, which would be updated according to \eqref{varphiUpdate} at the end of the considered scheduling period if
those particular $X_{ift}$ values are used. Constraints are the same as in the other schedulers. The integer
program in \eqref{pfsObj}-\eqref{pfsConstr} maximizes a nonlinear concave objective function subject to linear
constraints, where the concavity is due to the logarithm operation in the objective function.

To put it in a nutshell, our proposed scheduling model can be summarized as follows:

\emph{Step 1.} Find the values for $U_{if}$ $\forall i\in\mathcal{N}$, $\forall f\in\mathcal{F}$ by doing the analysis
in \eqref{initDeriv1}-\eqref{initDeriv4}.\\
\indent \emph{Step 2.} Given the values for $N$, $F$, $T$, $\varphi$, $\eta_{i}$, $a_{i}$, and $U_{if}$, find the values for $X_{ift}$ $\forall i\in\mathcal{N}$, $\forall f\in\mathcal{F}$,$\forall t\in\mathcal{T}$ by executing either of the following binary/mixed integer programs:\\
\indent \indent (a) For throughput maximizing scheduler, execute \eqref{TMSObj}-\eqref{TMSconstr3}\\
\indent \indent (b) For max-min fair scheduler, execute \eqref{mmfsObj}-\eqref{mmfsConstr2}\\
\indent \indent (c) For weighted max-min fair scheduler, execute \eqref{weightedMmfsObj}-\eqref{weightedMmfsConstr2}\\
\indent \indent (d) For proportionally fair scheduler, execute \eqref{pfsObj}-\eqref{pfsConstr}

Another strength of our work is that it is not only applicable to CRNs but to other wireless technologies as well. The integer programming formulations in Steps 2(a),(b),(c), and (d) above can be treated separately from Step 1; therefore, they can be applied for instance to OFDMA systems as well. This is not surprising; in fact, OFDMA is seen as a promising candidate technology for future CRNs \cite{mahmoud2009ofdm}.

All of the scheduling problems formulated in this paper are integer programming problems, which may in general be
NP-complete. However, some certain special cases of integer programming problems may be solvable in polynomial time.
For instance, the integer programming formulations for maximum weighted matching and minimum spanning tree problems
are solvable in polynomial time. In a separate study \cite{gozupek2011graph}, we proved that the throughput maximizing
scheduling problem is solvable in polynomial time, whereas the max-min fair, weighted max-min fair, and proportionally
fair scheduling problems are NP-complete in the strong sense. In this paper, we propose a computationally efficient
heuristic algorithm for the fair scheduling problems.

\subsection{Our Proposed Heuristic Algorithm}
In this section, we propose a heuristic algorithm for max-min, weighted max-min and proportionally fair schedulers. Our algorithm is adaptive in the sense that some of its steps behave differently depending on whether max-min fair, weighted max-min fair, or proportionally fair scheduling problem is under consideration.

We outline our proposed heuristic in Algorithm \ref{algHeuristic} (FAIRSCH).

\algsetup{indent=2em}
\newcommand{\factorial}{\ensuremath{\mbox{\sc Factorial}}}
\begin{algorithm}[h!]
\caption{FAIRSCH (Our proposed heuristic algorithm)}\label{algHeuristic}
\begin{algorithmic}[1]
\REQUIRE $\mathcal{N}$, $\mathcal{F}$, $\mathcal{T}$,
$\mathcal{A}_{i}$, $U_{if}$, $\eta_{i}$. \ENSURE $X_{ift}$ values $\forall
i\in\mathcal{N}, \forall f\in\mathcal{F},\forall t\in\mathcal{T}$.
\medskip

\STATE $\Omega_{i} \gets 0$, $\forall i\in\mathcal{N}$

\STATE $availSUs[t] \gets \mathcal{N}$ $\forall i\in\mathcal{N}$, $\forall t\in\mathcal{T}$

\IF {\textit{Weighted Max-Min}}

\STATE $U_{if} \gets \dfrac{U_{if}}{\eta_{i}}$, $\forall i\in\mathcal{N}$, $\forall f \in \mathcal{F}$
\ENDIF

\FOR{$f = 1$ to $F$}

\FOR{$t = 1$ to $T$}

\IF {\textit{Max-Min} \OR \textit{Weighted Max-Min}}

\STATE $i^{*} \gets \displaystyle\arg\min_{\overline{i}}\Omega_{\overline{i}}$, $\forall \overline{i} \in availSUs[t]$

\ELSE \IF {\textit{Prop-Fair}}

\STATE $newObj[\overline{i}] \gets (\Omega_{\overline{i}}+U_{\overline{i}f})\times \displaystyle\prod_{j \in \hspace{1mm}\mathcal{N} - \{\overline{i}\}}\Omega_{j}$, $\forall \overline{i} \in availSUs[t]$

\STATE $i^{*} \gets \displaystyle\arg\max_{\overline{i}}newObj[\overline{i}]$, $\forall \overline{i} \in availSUs[t]$

\ENDIF

\ENDIF

\STATE $X_{i^{*}ft} \gets 1$

\STATE $\Omega_{i^{*}} \gets \Omega_{i^{*}} + \dfrac{U_{i^{*}f}}{T}$

\IF {$\displaystyle\sum_{f=1}^{F}X_{i^{*}ft} = a_{i^{*}}$}

\STATE $availSUs[t] \gets availSUs[t]-\{i^{*}\}$

\ENDIF

\ENDFOR

\ENDFOR
\end{algorithmic}
\end{algorithm}

$\Omega_{i}$ indicates the summation of $U_{if}$ values that have hitherto been assigned to SU $i$ during the execution of
the algorithm. Step $1$ initializes all $\Omega_{i}$ values to $0$. $availSUs[t]$ indicates the list of SUs that have an
available (free) antenna for frequency assignment in time slot $t$. Since antennas of all SUs are available for all time slots at the
beginning of the execution of the algorithm, $availSUs[t]$ values are initialized to $\mathcal{N}$ in Step $2$. Steps
$3$-$5$ update the $U_{if}$ values for the weighted max-min fair scheduler since constraint \eqref{weightedMmfsConstr1} is
equivalent to \eqref{mmfsConstr1} when $U_{if}$ values are scaled to $\dfrac{U_{if}}{\eta_{i}}$. Algorithm \ref{algHeuristic}
assigns each frequency and time slot pair sequentially to an SU. For the max-min and weighted max-min fair schedulers,
Steps $8$-$9$ assign frequency $f$ and time slot $t$ to the SU that has the minimum $\Omega$ value so far among the SUs that have an
available antenna for time slot $t$ (ties are broken arbitrarily). For the proportionally fair scheduler, Steps $11$-$13$ assign frequency $f$ and
time slot $t$ to the SU that gives the maximum value for the product of $\Omega_{i}$ values if frequency $f$ and time slot $t$ are assigned. $newObj[i]$ indicates the
new objective function value, i.e., product of $\Omega_{i}$ values, if frequency $f$ and time slot $t$ are assigned to
SU $\overline{i}$. In the case of proportionally fair scheduling, our heuristic algorithm selects the SU that gives the
maximum value for this new objective function. Step $16$ makes the assignment to the selected SU $i^{*}$. If all antennas
of SU $i^{*}$ are assigned some frequency for time slot $t$, then steps $17$-$18$ remove this SU from the list of
available SUs for time slot $t$ ($availSUs[t]$). The algorithm terminates after all frequency and time slot pairs are assigned to
some SU. Note here that FAIRSCH is a greedy algorithm since it selects the SU that yields the best possible objective
function value in each iteration. In other words, it aims to increase the throughput of the SU with minimum (normalized)
throughput in (weighted) max-min fair scheduling, whereas it aims to maximize the product of the throughput values in
proportionally fair scheduling.

\textbf{Computational Complexity:} In the case of (weighted) max-min fair scheduling, FAIRSCH scans the list of SUs in $availSUs[t]$,
the size of which is at most $N$ during the assignment of frequency $f$ and time slot $t$. Since there are $F$ frequencies and $T$ time slots,
the complexity of FAIRSCH in the case of (weighted) max-min fair scheduling is $O(NFT)$. In the case of proportionally
fair scheduling, $availSUs[t]$ is scanned in the calculation of each value for $\overline{i}$ in Step $12$. Therefore,
the complexity of Step $12$ is $O(N^{2})$. Hence, the complexity of FAIRSCH in the case of proportionally fair scheduling
is $O(N^{2}FT)$.

\section{Simulation Results}\label{sec:SimResults}

We focus on a CRN cell with $600$ meters of radius and simulate it using Java. We elicit the $U_{if}$ values
in each set of simulations for $5000$ scheduling periods. We then solve the optimization problems using CPLEX
\cite{CPLEX} and KNITRO \cite{KNITRO}. We use CPLEX to solve the ILP formulations for throughput maximizing, max-min
fair, and weighted max-min fair schedulers. Since the optimization problem for the proportionally fair scheduler
in \eqref{pfsObj}-\eqref{pfsConstr} has a nonlinear objective function, we cannot solve it via CPLEX; therefore,
we solve it using KNITRO, which is a solver for nonlinear optimization \cite{KNITRO}. KNITRO uses two algorithms
for mixed integer nonlinear programming (MINLP). The first is a nonlinear branch-and-bound method and the second
implements the hybrid Quesada-Grossman method for convex MINLP \cite{quesada1992lp}. Because the objective function
in \eqref{pfsObj} is the maximization of a concave function, KNITRO can provide reliable solutions \cite{KNITRO}. We
compare the performance of our four schedulers by using the same set of $U_{if}$ values in each comparison. Every
scheduling period consists of $T = 10$ time slots and each time slot lasts for $T_s = 100$ milliseconds. According
to the IEEE $802.22$ standard, it is imperative that SUs vacate a spectrum band within two seconds from the appearance
of the licensed owner (PU) of that particular band. Therefore, we take each scheduling period equal to one second
(hence consisting of $10$ time slots) since it is sufficient for proper operation. Additionally, we examine our
methods for additive white gaussian noise (AWGN) channels; in other words, we take
$h_{ijt} = h_{it}^{CBS} = 1$ $\forall i \in \mathcal{N}$, $\forall j \in \Phi_{CBS}^{ft}$ and $\forall t\in\mathcal{T}$.
Furthermore, $\zeta=10^{-6}$, $\forall i,f$ and the maximum tolerable interference power of active PUs is
$P_{IF_{max}}^{fj} = 10$ milliwatts $\forall f\in\mathcal{F}$ and $\forall j\in \Phi_{CBS}^{ft}$.

Initial locations of SUs and PUs in the CRN cell are uniformly randomly distributed. We use
random waypoint mobility model to simulate the movement of SUs and PUs. In line with this model,
each SU/PU selects a target location in the cell in a uniformly random manner and moves towards this point with a
constant velocity. After reaching its target location, each node stays there for a certain amount of time and then
chooses another target point. This movement pattern continues in this way for each SU/PU until the end of simulation.
We set the staying duration between movement periods as $10$ seconds. We represent the velocity of SUs by $V_{s}$
and the velocity of PUs by $V_{p}$.

We model the PU spectrum utilization pattern by the finite state model illustrated in Figure \ref{fig:PUSpectrumOccupancy}. Each PU is either in the ON or in the OFF state. The ON state is comprised of one of the $F$ substates, each of which corresponds to being active using a frequency among a total of $F$ frequencies. The probability of staying in the ON or OFF states is $p_{S}$. The probability of selecting each frequency during switching from OFF to the ON state is equally likely; therefore, the probability of transition from OFF state to any frequency $f\in\mathcal{F}$ equals $(1-p_{S})/F$. In a slowly varying spectral environment, $p_{S}$ value is typically high; consequently, we select $p_{S}$ as $0.9$ in our simulations.

We identify six parameters and set possible low, high, and middle values for these parameters. We outline the names
of these parameters together with their low, middle, and high values in Table \ref{table:paramNames}. We take into
account the case where all SUs have the same velocity, referred to as $V_{s}$, and all PUs have the same
velocity, referred to as $V_{p}$. Moreover, we also focus on the case where all SUs have the same number of antennas;
i.e., $a_{i} = a$ $\forall i \in \mathcal{N}$.

By means of a $2^{6}$ factorial design (using the low and high values), multi-way (n-way) analysis of variation (ANOVA), and regression
techniques, we identified in \cite{Gozupek2010PIMRC} the following parameter pairs as statistically significant: $NM$, $NF$, $NV_p$, $MF$, $FV_p$, and $Fa$. In this paper, we analyze the impact of the interactions $NM$, $MF$, and $Fa$ in more detail. When we examine the influence of a parameter pair, we take the middle values of the rest of the parameters while varying the values of the examined parameter pair.

In both multi-way ANOVA and the detailed experiments, to determine the number of samples (replications, the number of scheduling periods) to run the experiments for, we employ the
following strategy: We initially obtain the $U_{if}$ values for $5000$ scheduling periods. We then take the first
$n+50$ of these samples (we will describe how to obtain the value for $n$ in the sequel), and run the optimization
problems in CPLEX/KNITRO for these sets of $U_{if}$ values. Our reason for this approach is twofold: Firstly because
CPLEX/KNITRO simulations may in some cases take considerable amount of time and hence we do not want to waste our
computational resources, and secondly because we want to ensure that our estimates for the response variables; i.e.,
total throughput for the throughput maximizing scheduler, minimum throughput for the max-min fair scheduler, normalized
minimum throughput for the weighted max-min fair scheduler, and the sum of logarithms of all SU throughput values for
the proportionally fair scheduler, are within a certain interval of the actual mean values with a high confidence level.
In particular, the number of samples we need to take hinges on the variance of the data. If the variance of the data
is small, there is no point in running the experiment with too many samples. In a data set where the variance is known,
the number of samples we need to take in order to ensure that the sample mean of the response variable is within
$\pm E$ of the actual mean with a $100(1 - \alpha)\%$ confidence level is as follows \cite{montgomery2007statistics}:

\begin{align}
n = \Bigg\lceil (\dfrac{z_{\alpha / 2}\sigma_{data}}{E})^{2}\Bigg\rceil \label{sampleSizeFormula}
\end{align}

Here, $z_{\alpha / 2}$ denotes the upper $\dfrac{\alpha}{2}\%$ percentile of the standard normal distribution, $n$
represents the sample size, and $\sigma_{data}$ symbolizes the standard deviation of the data, which corresponds to
the resulting value of the response variable. The ceiling operator in \eqref{sampleSizeFormula} is necessary because
$n$ has to take integer values. Here, $2E$ denotes the width of the confidence interval (CI). In our experiments, we
take $\alpha = 0.05$ and $E = 0.5$; in other words, we can say with $95\%$ confidence level that we are within
$\pm 0.5$ of the actual mean of the response variable in our experiments.

An important point that needs to be taken into account in our case is that we do not know the actual standard
deviation of the data (response variable). Hence, we statistically estimate $\sigma_{data}$ and insert our estimated
value in the formula in \eqref{sampleSizeFormula}, thereby following an iterative methodology. We initially take $50$
samples since the central limit theorem prescribes that at least $40$ samples should be taken in order for the formula
in \eqref{sampleSizeFormula} to hold true. We calculate the standard deviation of these data with $50$ samples and
find the value for $n$ using \eqref{sampleSizeFormula}. We take an additional set of $n$ samples and calculate the
sample mean of the response variable in these additional $n + 50$ samples. Finally, we conclude that this estimate is
our final estimate for the sample mean of the response variable in our data. In order to verify the validity of our
method, we calculate the standard deviation of these $n + 50$ samples and find another value for $n$ via
\eqref{sampleSizeFormula}, which we call $n_{new}$. If $n_{new} > n + 50$, where $n + 50$ is our actual sample size,
it implies that there is an undesired feature associated with the data, for instance the samples not being independent
of each other. We checked this condition for all of our experiments and did not observe this undesired behavior in any
of them. Therefore, we verified the validity of our estimation procedure for the mean of the response variable.

We present in Figure \ref{fig:factorNM-2D-EN} the results for the $NM$ interaction for the throughput maximizing
scheduler; how the average total network throughput is affected by varying $N$ and $M$. We can see that the total
network throughput decreases as the number of PUs in the CRN cell ($M$) increases. Increasing the number of PUs
implies a decrease in the transmission power of the SUs and consequently their data rate in order not to disturb
these PUs. The increase in the total network throughput as the number of SUs ($N$) increases, on the other hand,
can be attributed to the opportunistic nature of the throughput maximizing scheduler. A frequency and time slot pair
can be regarded as a resource that needs to be assigned to only one SU. Because this resource is most of the time
(except in order to comply with the time slot constraint \eqref{TMSconstr1} and antenna constraint \eqref{TMSconstr3})
assigned to the SU that has the best channel conditions (the maximum $U_{if}$ value), the probability that an SU with
better channel conditions exists increases as the number of SUs increases.

Figure \ref{fig:factorMF-2D-EN} shows the results for the MF interaction for the throughput maximizing scheduler.
Because of the same reasoning as in Figure \ref{fig:factorNM-2D-EN}, the total average network throughput decreases
also here as the number of PUs in the CRN cell increases. Moreover, we can observe that the total network throughput
increases almost linearly with the number of frequencies ($F$) in the CRN cell. This observation is intuitive because
the more frequencies there are in the CRN cell, the more resources there are for the SUs to send their data through.
Furthermore, we can observe that the number of frequencies is a very important factor on the network throughput. In
fact, it is even more important than the number of PUs in the CRN cell, since decreasing the number of frequencies in
the CRN cell has more impact on decreasing the network throughput in comparison to increasing the number of PUs in the
CRN cell.

We show in Figure \ref{fig:factorFa-2D-EN} the results for the $Fa$ interaction for the throughput maximizing scheduler.
We see that increasing the number of antennas of the SUs makes sense only when there is a certain number of frequencies
in the system. For instance, having one antenna has almost the same effect as having multiple antennas when the number
of frequencies is less than $15$. On the other hand, when the number of frequencies is between $15$ and $30$,
increasing the number of antennas from $1$ to $2$ makes a significant difference; nevertheless, having more than two
antennas still does not make sense. We see a similar behavior at $F=30$; i.e., when each SU has $2$ antennas, having
$F>30$ does not increase the total throughput. The reason for this behavior is constraint \eqref{TMSconstr1}, which
ensures that each SU is assigned at least one time slot. In order to comply with this constraint, the scheduler tends
to initially assign some frequency to the first antenna of each SU and then continue assigning frequencies to the other
antennas. Recall that $N = 15$ in Figure \ref{fig:factorFa-2D-EN}. Until the point where $N = F$, the scheduler assigns
the frequencies to the first antennas of each SU. Increasing the number of antennas has a similar effect on total throughput
as increasing the number of SUs. Hence, between $F=N$ and $F = 2N$, the scheduler tends to assign the frequencies to
the second antennas of each SU. Bear in mind that this is the average case behavior of the scheduler. Taking
this scheduler behavior into consideration, we conclude that a centralized network entity responsible for assigning
frequency bands to several CBSs may opt not to assign more frequencies ($F$) than $aN$ to a particular CBS since it
does not bring any additional total throughput advantage to this CRN cell. Likewise, the decision about how many
antennas the SUs should have can be made by considering the number of frequencies ($F$) in the CRN cell, since
having multiple antennas has an additional hardware cost.

Figure \ref{fig:FairSch_BarGraph} displays the behavior of throughput maximizing scheduler (\emph{Thr-Max}), max-min
fair scheduler (\emph{Max-Min}), weighted max-min fair scheduler (\emph{Weighted Max-Min}), and proportionally fair
scheduler (\emph{Prop-Fair}) for $N = 5$ SUs, window size ($\varphi$) = $5$, where the other parameters
($M$, $F$, $V_{p}$, $V_{s}$, and $a$) are set to their middle values, as shown in Table \ref{table:paramNames}.
Target weights for the weighted max-min fair scheduler are
$\eta_{1} = 0.05$, $\eta_{2} = 0.1$, $\eta_{3} = 0.2$, $\eta_{4} = 0.25$, $\eta_{5} = 0.4$. We can see here
that when compared with the throughput maximizing scheduler, max-min fair scheduler achieves more uniform
throughput among the SUs at the expense of a small decrease in the total network throughput, hence achieving
fairness. The throughput values
achieved by the weighted max-min fair scheduler, on the other hand, are in line with their target weights at the
expense of less total network throughput than the other two schedulers. Besides, proportionally fair scheduler also
results in similar throughput values among the SUs; however, the resulting throughput values are not as close to
each other as in the max-min fair scheduler. Throughput values of each SU are closer to the throughput values in
the throughput maximizing scheduler. In other words, proportionally fair scheduler, as expected, exhibits a tradeoff
between users' satisfaction and system revenue.

We run a similar experiment with window size ($\varphi$) = 10 and observe that the total throughput achieved by the
weighted max-min fair scheduler with $\varphi = 10$ is more than the one achieved by the weighted max-min fair
scheduler with $\varphi = 5$. We also observe that this increase in total throughput is accomplished at the expense
of a small deviation from the target weights. In contrast, performances of max-min fair scheduler and proportionally
fair scheduler do not change much with varying window size. Therefore, we have decided to analyze the impact of window size on
these schedulers in more detail. We do not explicitly show the results of the experiment with $\varphi = 10$ in a
separate figure since we show in the sequel the results of this experiment as part of all the results in Figure
\ref{fig:FairSch_VaryingWindowSize} and Table \ref{table:weightedMMFS}.

In Table \ref{table:weightedMMFS}, we show the resulting throughput ratios of each SU for varying window size.
Apparently, the deviation from target weights increases as the window size increases. When we compare the cases
where $\varphi = 1$ and $\varphi = 50$, we can see that the deviation from target weights can become quite large.

Jain's index is a commonly used fairness index in the literature \cite{jain1999throughput}. It is calculated as
$\dfrac{(\sum_{i=1}^{N}\Omega_{i})^{2}}{N \times \sum_{i=1}^{N}\Omega_{i}^{2}}$, where $\Omega_{i}$ is the throughput
of SU $i$ and $N$ is the total number of SUs. Jain's index becomes closer to $1$ as the throughput values of the SUs become
closer to each other. The minimum and maximum values that it can take are $\dfrac{1}{N}$ and $1$, respectively
\cite{jain1999throughput}. We present in Figure \ref{fig:FairSch_5SUs_VaryingWindowSize_TotThr} the total average
throughput values and in Figure \ref{fig:FairSch_5SUs_WindowSize5_Jain} the average Jain fairness index values for
all schedulers with varying window size ($\varphi$). While calculating the average Jain index values, we have taken the arithmetic mean of the Jain index values calculated after every scheduling period. The performance of throughput maximizing scheduler in terms of
both criteria naturally remains constant since window size is not a parameter of this scheduler. The total throughput
and Jain index performances of max-min fair and proportionally fair schedulers, on the other hand, are almost invariant
of the window size. However, the performance of the weighted max-min fair scheduler is significantly affected by the
window size in terms of both performance criteria. The total throughput increases as the window size increases until
it saturates at some point and becomes close to the total throughput of max-min fair scheduler. Hence, we observe that
our windowing mechanism provides our weighted max-min fair scheduler with the flexibility to provide a tradeoff between
maximizing total throughput and adhering to the target throughput proportions.

Throughput and fairness results in Figure \ref{fig:FairSch_VaryingWindowSize} call for weighted max-min fair
schedulers that determine the optimal window size ($\varphi$) achieving a certain throughput or fairness objective.
First, in line with our observations in Figure \ref{fig:FairSch_5SUs_VaryingWindowSize_TotThr}, we can formulate the
following scheduling problem:
\begin{align}
\max\hspace{1mm} &Z' \label{weightedMmfsNew1Obj}\\
s.&t.\notag \\
\sum_{i=1}^{N}\sum_{f=1}^{F}&\sum_{t=1}^{T}\dfrac{U_{if}X_{ift}}{T} \ge \Omega \label{weightedMmfsNew1Constr1}\\
\eqref{weightedMmfsConstr1},& \text{ and } \eqref{weightedMmfsConstr2}\label{weightedMmfsNew1Constr2}
\end{align}
where $\Omega$ is the desired minimum total throughput value, which is fed as an input variable to the optimization
problem in \eqref{weightedMmfsNew1Obj}-\eqref{weightedMmfsNew1Constr2}. Unlike in \eqref{weightedMmfsObj}-\eqref{weightedMmfsConstr2},
the variable $\varphi$ is a decision variable rather than an input variable.

Second, in line with our observations in Figure \ref{fig:FairSch_5SUs_WindowSize5_Jain}, we can formulate the following
scheduling problem:

\begin{align}
&\max\hspace{1mm} Z' \label{weightedMmfsNew2Obj}\\
&s.t.\notag \\
&\dfrac{(\displaystyle\sum_{i=1}^{N}\sum_{f=1}^{F}\sum_{t=1}^{T}\dfrac{U_{if}X_{ift}}{T})^{2}}{N \times \displaystyle\sum_{i=1}^{N}(\sum_{f=1}^{F}\sum_{t=1}^{T}\dfrac{U_{if}X_{ift}}{T})^{2}} \ge J \label{weightedMmfsNew2Constr1}\\
&\eqref{weightedMmfsConstr1}, \text{ and }\eqref{weightedMmfsConstr2}\label{weightedMmfsNew2Constr2}
\end{align}

where $J$ is the desired Jain's fairness index value, which is fed as an input variable to the optimization
problem in \eqref{weightedMmfsNew2Obj}-\eqref{weightedMmfsNew2Constr2}. Again unlike in
\eqref{weightedMmfsObj}-\eqref{weightedMmfsConstr2}, the variable $\varphi$ is a decision variable rather than an input variable. Both problems in
\eqref{weightedMmfsNew1Obj}-\eqref{weightedMmfsNew1Constr2} and \eqref{weightedMmfsNew2Obj}-\eqref{weightedMmfsNew2Constr2}
are nonlinear integer programming problems; therefore, they are computationally difficult. Furthermore,
their nonlinearity prevents the usage of optimization software such as CPLEX. Finding efficient algorithms to address
these problems is a research challenge, which is left as future work.

We present in Figure \ref{fig:FairSch_VaryingNumSUs} the average total throughput and average Jain index values of
throughput maximizing, max-min fair, and proportionally fair schedulers for $\varphi = 5$ and varying number of
secondary users, where the other parameters ($M$, $F$, $V_p$, $V_s$, $a$) are set to their middle values. We do not
present the values for the weighted max-min fair scheduler here because the target weights of each SU have to change
as the number of SUs increases and these different simulation scenarios cannot be compared when there are different
target weights. Furthermore, target weights become very small as the number of SUs increases and it becomes
difficult to truly assess the throughput and fairness performance of the weighted max-min fair scheduler. Recall
that the throughput maximizing scheduler assigns the resources (frequencies and time slots) most of the time to the
SU that has the best channel conditions and the probability that an SU with better channel conditions exists increases
as the number of SUs in the CRN cell increases. This multiuser diversity created by the opportunistic behavior of the
throughput maximizing scheduler increases the variation between the SU throughput values and hence decreases Jain's
fairness index as the number of SUs increases. Max-min fair scheduler exhibits the highest Jain's fairness
index and the lowest total throughput since its objective is to make the throughput values as close to each other
as possible. Proportionally fair scheduler, on the other hand, again displays a tradeoff between the overall system
revenue (total throughput) and individual throughput values. Its total throughput and Jain's index are between the
corresponding values of the throughput maximizing scheduler and the max-min fair scheduler.

Figure \ref{fig:HeuristicMMFS_VaryingNumSUs} shows the average minimum throughput (objective function value) of
CPLEX results (Max-Min) and our heuristic algorithm (Heuristic-MMFS) for $M=20$, $F=15$, $V_{p}=V_{s}=13$ m/s, $a_{i}=3$
$\forall i\in\mathcal{N}$, $\varphi=5$, and varying number of SUs. Our proposed heuristic yields close performance
to the values obtained from CPLEX. Furthermore, we also observe that average minimum throughput decreases as the number
of SUs increases since the amount of resources that each SU can receive decreases as there are more SUs
competing for the same amount of resources. Figure \ref{fig:HeuristicPFS_VaryingNumSUs} shows the sum of logarithms of
SU throughput values (objective function value) of KNITRO results (Prop-Fair) and our heuristic algorithm
(Heuristic-PFS) for the same parameters as in Figure \ref{fig:HeuristicMMFS_VaryingNumSUs}. Results demonstrate that
our proposed heuristic yields very close performance to the values obtained from KNITRO. Objective function value increases
as the number of SUs increases; however, they saturate at some point. This behavior is due to the logarithm in the
objective function value.

Figure \ref{fig:HeuristicMMFS_VaryingNumFreqsN5} shows the average minimum throughput (objective function value) of
CPLEX results (Max-Min) and our heuristic algorithm (Heuristic-MMFS) for $N=5$, $M=20$, $V_{p}=V_{s}=13$ m/s,
$a_{i}=3$ $\forall i\in\mathcal{N}$, $\varphi=5$, and varying number of frequencies. Our proposed heuristic yields close
performance to the values obtained from CPLEX. Furthermore, we also observe that average minimum throughput in general
increases as $F$ increases since there are more resources to be shared by the same number of SUs. Figure \ref{fig:HeuristicWeightedMMFS_VaryingNumFreqsN5} shows the normalized
average minimum throughput (objective function value) of CPLEX results (Weighted Max-Min) and our heuristic algorithm
(Heuristic-Weighted-MMFS) for the same parameters, where normalized average minimum throughput indicates the minimum
throughput obtained after the $U_{if}$ values are updated as $\dfrac{U_{if}}{\eta_{i}}$ in Steps $3$-$5$ of Algorithm
\ref{algHeuristic}. Results demonstrate that our heuristic algorithm achieves close results to the ones obtained
from CPLEX. Figure \ref{fig:HeuristicPFS_VaryingNumFreqsN5} shows the sum of logarithms of
SU throughput values (objective function value) of KNITRO results (Prop-Fair) and our heuristic algorithm
(Heuristic-PFS) for the same parameters. Results demonstrate that
our proposed heuristic yields very close performance to the values obtained from KNITRO. Objective function value
increases as the number of frequencies increases; however, they saturate at some point. This behavior is due to the logarithm in the
objective function value.

Figure \ref{fig:HeuristicVaryingNumFreqsN15} shows the average minimum throughput (objective function value) of
CPLEX results (Max-Min) and our heuristic algorithm (Heuristic-MMFS) for $N=15$, $M=20$, $V_{p}=V_{s}=13$ m/s, $a_{i}=3$
$\forall i\in\mathcal{N}$, $\varphi=5$, and varying number of frequencies. Our proposed heuristic yields close performance
to the values obtained from CPLEX. Figure \ref{fig:HeuristicPFS_VaryingNumFreqsN15} shows the sum of logarithms of
SU throughput values (objective function value) of KNITRO results (Prop-Fair) and our heuristic algorithm
(Heuristic-PFS) for the same parameters. Results demonstrate that
our proposed heuristic yields very close performance to the values obtained from KNITRO.

\section{Conclusion}\label{sec:Conc}
In this paper, we have formulated throughput maximizing, max-min fair, weighted max-min fair, and proportionally
fair scheduling problems for cognitive radio networks managed by a centralized cognitive base station. For the fair scheduler, we proposed an adaptive heuristic algorithm, which acts differently depending on whether max-min fair, weighted max-min fair or proportionally fair scheduling problem is under consideration. The distinguishing feature
of our scheduling model is that it is a a very general model jointly accomplishing numerous goals such as making
the frequency, time slot, and data rate allocation to the secondary (cognitive) users, which possibly have multiple antennas, in a heterogenous
multi-channel and multi-user scenario. Common features of all schedulers are that all of them ensure that the primary
users in the service area of the cognitive base station are not disturbed, no collisions occur among the secondary
users, and reliable communication with the cognitive base station is maintained. One of the distinctive features of our
fair schedulers is that they take into account the throughput performance of the secondary users in the recent past through a windowing mechanism and utilize this information to
make the scheduling decision in the current scheduling period. This mechanism provides our schedulers with the ability
to compensate for the possible temporary throughput losses of the secondary users in the subsequent scheduling periods.
Moreover, our fair schedulers also have the property of providing joint temporal and throughput fairness.

We have assessed the performance of the throughput maximizing
scheduler with respect to various parameters such as the number of secondary users, primary users, frequencies, and
antennas. In addition, we have made a comparative evaluation of all schedulers in terms of average total throughput
and Jain's fairness index for varying window size and varying number of secondary users. We have observed that
increasing the window size does not change the performance of our max-min and proportionally
fair schedulers; however, it increases the average total throughput in the weighted max-min fair scheduler at the
expense of an increase in the deviation from target weights. We have demonstrated
through simulations that our proposed heuristic algorithm for fair schedulers yields close performance to the solutions obtained from optimization softwares CPLEX and KNITRO.

The set of scheduling algorithms we propose in this paper can be readily used by the cognitive base station (CBS) operator in a dynamic and adaptive manner. Hence, our proposed algorithms can be regarded as part of a scheduling model for a CBS operator. Our findings indicate that each of our schedulers work better under
different conditions. Scheduling model at the CBS can dynamically change the executed
scheduler for each scheduling period according to the network conditions and requirements. In particular, our scheduling model can implement the following:

\begin{itemize}
\item If the number of secondary users is small, the model can execute throughput maximizing scheduler because the simulation results indicate that the fairness index of throughput maximizing scheduler is only slightly lower than the ones of fair schedulers. Therefore, it does not make sense to sacrifice from total throughput when there is a small number of secondary users (SU).
\item If the number of SUs is large, the CBS operator gives importance to fairness, there is no priority difference between SUs, and the spectrum and channel conditions between SUs is fairly uniform, then the model can execute max-min fair scheduler (MMFS).
\item If the number of SUs is large, the CBS operator gives importance to fairness, there is no priority difference between SUs, and the spectrum and channel conditions between SUs is highly heterogeneous, then the scheduling model can execute proportionally fair scheduler (PFS). The reason for this decision is that if MMFS is executed, an SU with very bad channel conditions can drive the total throughput of all SUs to very low values. PFS scheduler provides a good tradeoff between maximizing total throughput and achieving fairness.
\item If the number of SUs is large, the CBS operator gives importance to fairness, and there is priority difference between SUs, then the scheduling model can execute weighted MMFS scheduler in order to provide service differentiation capability to the SUs.
\end{itemize}

As a future work, we plan to design approximation algorithms, which have theoretically
provable performance guarantee, to address the fair scheduling problems.




\bibliographystyle{wileyj}

\begin{thebibliography}{10}
\expandafter\ifx\csname url\endcsname\relax
  \def\url#1{\texttt{#1}}\fi
\expandafter\ifx\csname urlprefix\endcsname\relax\def\urlprefix{URL }\fi
\expandafter\ifx\csname href\endcsname\relax
  \def\href#1#2{#2} \def\path#1{#1}\fi

\bibitem{Mitola}
J.~Mitola, {Cognitive Radio: An Integrated Agent Architecture for Software
  Defined Radio}, PhD Thesis, Royal Institute of Technology (KTH), Stockholm,
  Sweden.

\bibitem{akyildiz2006next}
I.~Akyildiz, W.~Lee, M.~Vuran, S.~Mohanty, Next generation/dynamic spectrum
  access/cognitive radio wireless networks: a survey, Computer Networks 50~(13)
  (2006) 2127--2159.

\bibitem{gozupek2009jcn}
D.~G\"{o}z\"{u}pek, F.~Alag\"{o}z, Throughput and delay optimal scheduling in
  cognitive radio networks under interference temperature constraints, Journal
  of Communications and Networks 11~(2) (2009) 147--155.

\bibitem{gozupek2011genetic}
D.~G\"{o}z\"{u}pek, F.~Alag\"{o}z, Genetic algorithm-based scheduling in
  cognitive radio networks under interference temperature constraints,
  International Journal of Communication Systems 24~(2) (2011) 239--257.

\bibitem{fcc2003docket}
E.~FCC, Docket no 03-222 notice of proposed rule making and order (2003).

\bibitem{order2007fcc}
F.~Order, Fcc-07-78a1 (2007).

\bibitem{quan2008optimal}
Z.~Quan, S.~Cui, A.~Sayed, Optimal linear cooperation for spectrum sensing in
  cognitive radio networks, IEEE Journal on Selected Topics in Signal
  Processing 2~(1) (2008) 28--40.

\bibitem{jiang2009optimal}
H.~Jiang, L.~Lai, R.~Fan, H.~Poor, Optimal selection of channel sensing order
  in cognitive radio, IEEE Transactions on Wireless Communications 8~(1) (2009)
  297--307.

\bibitem{liu2001opportunistic}
X.~Liu, E.~Chong, N.~Shroff, Opportunistic transmission scheduling with
  resource-sharing constraints in wireless networks, IEEE Journal on Selected
  Areas in Communications 19~(10) (2001) 2053--2064.

\bibitem{viswanath2002opportunistic}
P.~Viswanath, D.~Tse, R.~Laroia, Opportunistic beamforming using dumb antennas,
  IEEE Transactions on Information Theory 48~(6) (2002) 1277--1294.

\bibitem{liu2003framework}
X.~Liu, E.~Chong, N.~Shroff, A framework for opportunistic scheduling in
  wireless networks, Computer Networks 41~(4) (2003) 451--474.

\bibitem{ye2011ijcs}
J.~Ye, J.X.~Wang, J.W.~Huang, A cross-layer TCP for providing fairness in wireless mesh networks, Wiley's International Journal of Communication Systems, vol.~24, no.~12, pp.1611--1626, 2011.

\bibitem{hou2011ijcs}
T.C.~Hou, C.W.~Hsu, C.S.~Wu, A delay-based transport layer mechanism for fair TCP throughput over 802.11 multihop wireless mesh networks, Wiley's International Journal of Communication Systems, vol.~24, no.~8, pp.1015--1032, 2011.

\bibitem{yang2012study}
F.M.~Yang, W.M.~Chen, T.K.~Cheng, J.L.~Wu, A study of QoS guarantee and fairness based on cross-layer channel state in Worldwide Interoperability for Microwave Access, Wiley's International Journal of Communication Systems, vol.~25, no.~7, pp.926--942, 2012.

\bibitem{huang2001max}
X.~Huang, B.~Bensaou, On max-min fairness and scheduling in wireless ad-hoc
  networks: Analytical framework and implementation, ACM International
  Symposium on Mobile Ad Hoc Networking \& Computing, 2001, pp. 221--231.

\bibitem{liu2003opportunistic}
Y.~Liu, E.~Knightly, Opportunistic fair scheduling over multiple wireless
  channels, IEEE International Conference on Computer Communications (INFOCOM), Vol.~2, 2003, pp. 1106--1115.

\bibitem{nandagopal2000achieving}
T.~Nandagopal, T.~Kim, X.~Gao, V.~Bharghavan, Achieving mac layer fairness in
  wireless packet networks, ACM MOBICOM, 2000, pp. 87--98.

\bibitem{tang2008joint}
J.~Tang, S.~Misra, G.~Xue, Joint spectrum allocation and scheduling for fair
  spectrum sharing in cognitive radio wireless networks, Computer Networks
  52~(11) (2008) 2148--2158.

\bibitem{lee2010spectrum}
W.~Lee, I.~Akyildiz, A spectrum decision framework for cognitive radio
  networks, IEEE Transactions on Mobile Computing (2010) 161--174.

\bibitem{lee2008joint}
W.~Lee, I.~Akyildiz, Joint spectrum and power allocation for inter-cell
  spectrum sharing in cognitive radio networks, IEEE Symposium on New Frontiers in Dynamic Spectrum Access Networks (DySpan), 2008.

\bibitem{nguyen2011effective}
M.~Nguyen, H.~Lee, Effective scheduling in infrastructure-based cognitive radio
  network, IEEE Transactions on Mobile Computing 10~(6) (2011) 853--867.

\bibitem{shi2010maximizing}
Y.~Shi, Y.~Hou, S.~Kompella, H.~Sherali, Maximizing capacity in multihop
  cognitive radio networks under the sinr model, IEEE Transactions on Mobile
  Computing (2010) 954--967.

\bibitem{peng2006utilization}
C.~Peng, H.~Zheng, B.~Zhao, Utilization and fairness in spectrum assignment for
  opportunistic spectrum access, Mobile Networks and Applications 11~(4) (2006)
  555--576.

\bibitem{dianati2007scheduling}
M.~Dianati, X.S.~Shen, K.~Naik, Scheduling with base station diversity and fairness analysis for the downlink of CDMA cellular networks, Wiley Journal on Wireless Communications and Mobile Computing 7~(5) (2007) 569--579.

\bibitem{order2008fcc}
F.~Order, Fcc-08-260 (2008).

\bibitem{wang2009effects}
L.~Wang, A.~Chen, Effects of location awareness on concurrent transmissions for
  cognitive ad hoc networks overlaying infrastructure-based systems, IEEE
  Transactions on Mobile Computing (2009) 577--589.

\bibitem{wang2010resource}
H.~Wang, J.~Ren, T.~Li, Resource allocation with load balancing for cogntive
  radio networks, IEEE Global Communications Conference (GLOBECOM), 2010.

\bibitem{digham2008joint}
F.~Digham, Joint power and channel allocation for cognitive radios, IEEE
  Wireless Communications and Networking Conference (WCNC), 2008, pp.
  882--887.

\bibitem{cover2006elements}
T.~Cover, J.~Thomas, Elements of information theory, Wiley-Interscience, New
  York, 2006.

\bibitem{rajan2004delay}
D.~Rajan, A.~Sabharwal, B.~Aazhang, Delay-bounded packet scheduling of bursty
  traffic over wireless channels, IEEE Transactions on Information Theory
  50~(1) (2004) 125--144.

\bibitem{mahmoud2009ofdm}
H.~Mahmoud, T.~Yucek, H.~Arslan, OFDM for cognitive radio: Merits and challenges, IEEE Wireless Communications 6~(2) (2009) 6--15.

\bibitem{li2008proportional}
L.~Li, M.~Pal, Y.~Yang, Proportional fairness in multi-rate wireless lans,
  IEEE International Conference on Computer Communications (INFOCOM), 2008, pp. 1004--1012.

\bibitem{gozupek2011graph}
D.~G\"{o}z\"{u}pek, M.~Shalom, F.~Alag\"{o}z, A graph theoretic approach to
  scheduling in cognitive radio networks,  \emph{submitted to IEEE/ACM
  Transactions on Networking}, 2012. [Online]. Available:
  \url{http://www.cmpe.boun.edu.tr/\urltilda gozupek/gozupek-GraphCogSch-2012.pdf}

\bibitem{ileri2005demand}
O.~Ileri, D.~Samardzija, N.B.~Mandayam, Demand responsive pricing and competitive spectrum allocation via a spectrum server, \emph{IEEE DySPAN}, pp.194-202, 2005.

\bibitem{CPLEX}
CPLEX, [Online]. Available:
  \url{http://www.ilog.com/products/cplex}

\bibitem{KNITRO}
{KNITRO Solver, Ziena Optimization Inc.}, [Online]. Available:\url{http://www.ziena.com/knitro.htm}

\bibitem{quesada1992lp}
I.~Quesada, I.~Grossmann, An lp/nlp based branch and bound algorithm for convex
  minlp optimization problems, Computers \& Chemical Engineering 16~(10-11)
  (1992) 937--947.

\bibitem{Gozupek2010PIMRC}
D.~G\"{o}z\"{u}pek, F.~Alag\"{o}z, {An interference aware throughput maximizing
  scheduler for centralized cognitive radio networks}, IEEE International
  Symposium on Personal, Indoor and Mobile Radio Communications (PIMRC),
  2010, pp. 1--6.

\bibitem{montgomery2007statistics}
D.~Montgomery, G.~Runger, Applied statistics and probability for engineers,
  Wiley India Pvt.Ltd., 2007.

\bibitem{jain1999throughput}
R.~Jain, A.~Durresi, G.~Babic, Throughput fairness index: an explanation, in:
  ATM Forum Contribution, Vol.~45, 1999.

\bibitem{rodrigues2009adaptive}
E.B.~Rodrigues, F.~Casadevall, {Adaptive radio resource allocation framework for multi-user OFDM}, IEEE Vehicular Technology Conference (VTC),
  2009, pp. 1--6.

\bibitem{zhang2008opportunistic}
Z.~Zhang, Y.~He, E.K.P.~Chong, Opportunistic scheduling for OFDM systems with fairness constraints, EURASIP Journal on Wireless Communications and Networking (2008).

\bibitem{wang2012resource}
S.~Wang, F.~Huang, M.~Yuan, S.~Du, Resource allocation for multiuser cognitive OFDM networks with proportional rate constraints, Wiley's International Journal of Communication Systems, vol.~25, no.~2, pp.254--269, 2012.

\bibitem{derya2012wcmc}
D.~Cavdar, H.B.~Yilmaz, T.~Tugcu, F.~Alagöz, Analytical modeling and resource planning for cognitive radio systems, Wiley Journal on Wireless Communications and Mobile Computing, vol.~12, no.~3, pp.277--292, 2012.

\end{thebibliography}

\newpage

\begin{figure}[h]
\begin{center}
\includegraphics[scale=0.55]{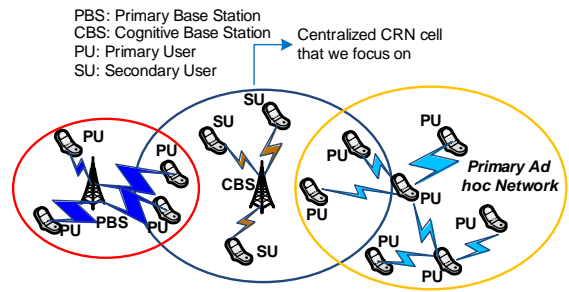}
\caption{The considered centralized CRN architecture.} \label{fig:802_22_NetworkDiagram}
\end{center}
\end{figure}
\FloatBarrier

\begin{figure}[h]
\begin{center}
\includegraphics[scale=0.6]{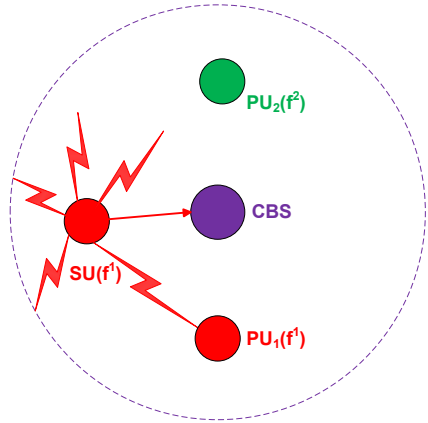}
\caption{Framework for our cognitive scheduling mechanism.} \label{fig:OurSchedulingMechanism}
\end{center}
\end{figure}
\FloatBarrier

\begin{figure}[h]
\begin{center}
\includegraphics[scale=0.6]{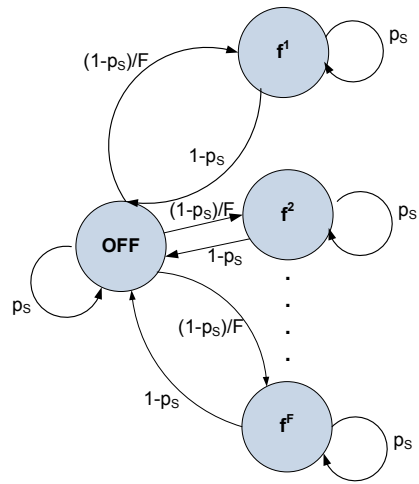}
\caption{PU spectrum occupancy model.} \label{fig:PUSpectrumOccupancy}
\end{center}
\end{figure}
\FloatBarrier

\begin{table}
  \centering
\begin{center}
\small
  \begin{tabular}{ | c | c | c | c | }
    \hline
    Parameter Name & Low Value & High Value & Middle Value \\ \hline
    N (Number of SUs) & 5 & 30 & 15\\ \hline
    M (Number of PUs) & 5 & 40 & 20\\ \hline
    F (Number of frequencies) & 3 & 30 & 15 \\ \hline
    $V_{p}$ (Velocity of PUs) & 1 m/s & 25 m/s & 13 m/s \\ \hline
     $V_{s}$ (Velocity of SUs) & 1 m/s & 25 m/s & 13 m/s \\ \hline
    $a$ (Number of antennas of SUs) & 1 & 5 & 3 \\ \hline
  \end{tabular}
\end{center}
  \caption{Parameter names and low/high/middle values}\label{table:paramNames}
\end{table}
\FloatBarrier

\begin{figure}[h]
\begin{center}
\includegraphics[scale=0.55]{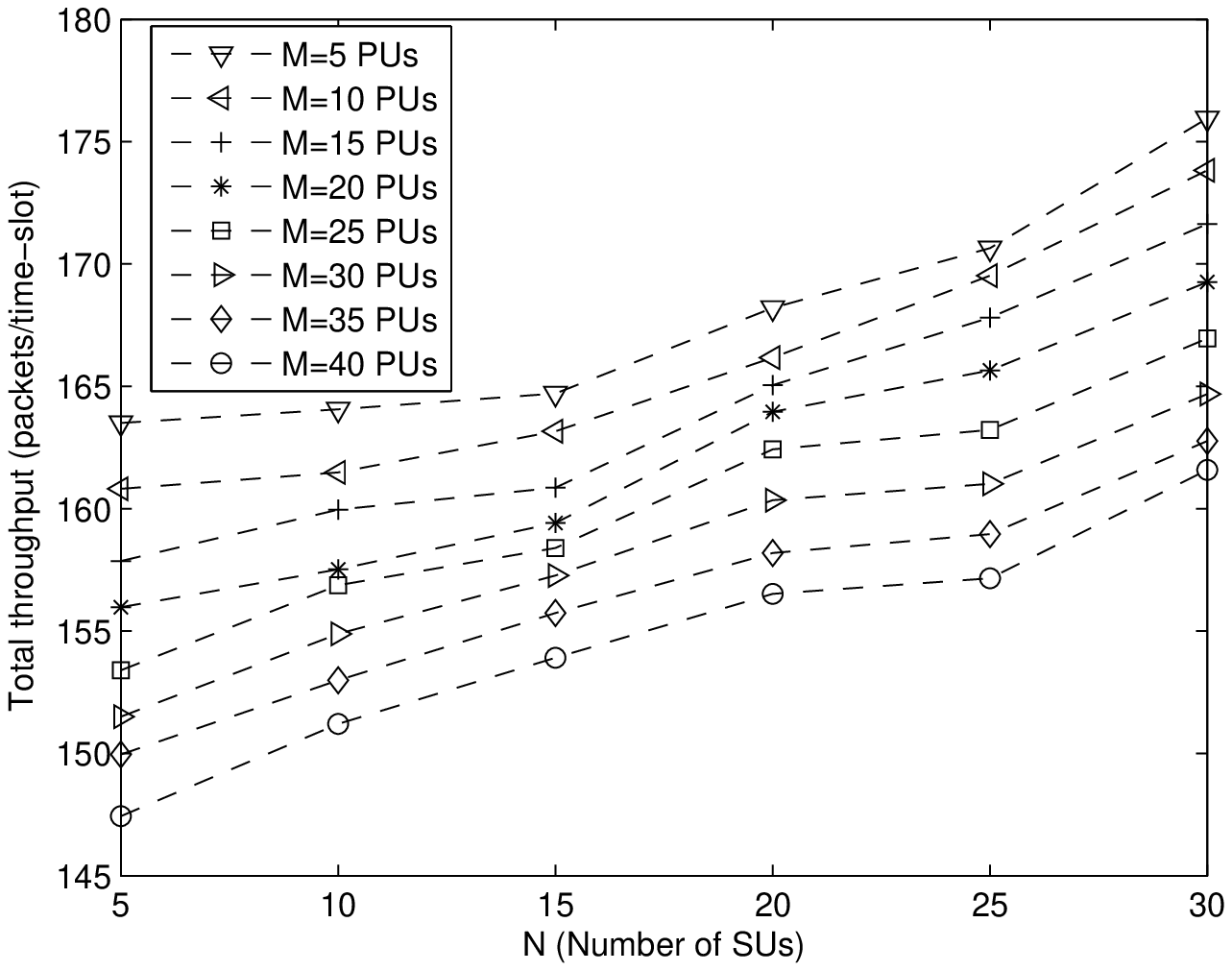}
\caption{Average total network throughput for throughput maximizing scheduler for varying $N$ and $M$.} \label{fig:factorNM-2D-EN}
\end{center}
\end{figure}
\FloatBarrier

\begin{figure}[h]
\begin{center}
\includegraphics[scale=0.55]{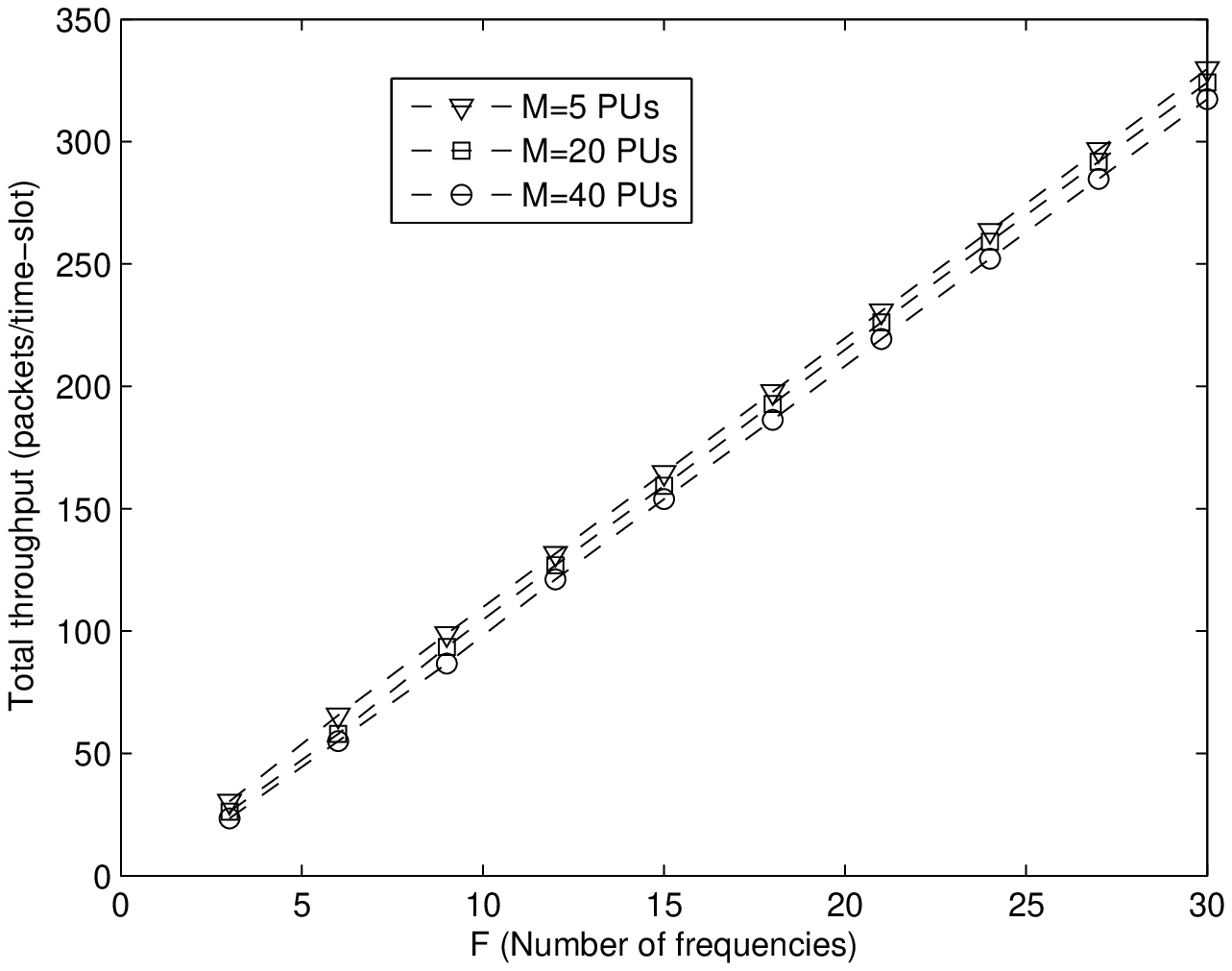}
\caption{Average total network throughput for throughput maximizing scheduler for varying $N$ and $F$.} \label{fig:factorMF-2D-EN}
\end{center}
\end{figure}
\FloatBarrier

\newpage
\begin{figure}[!htb]
\begin{center}
\includegraphics[scale=0.55]{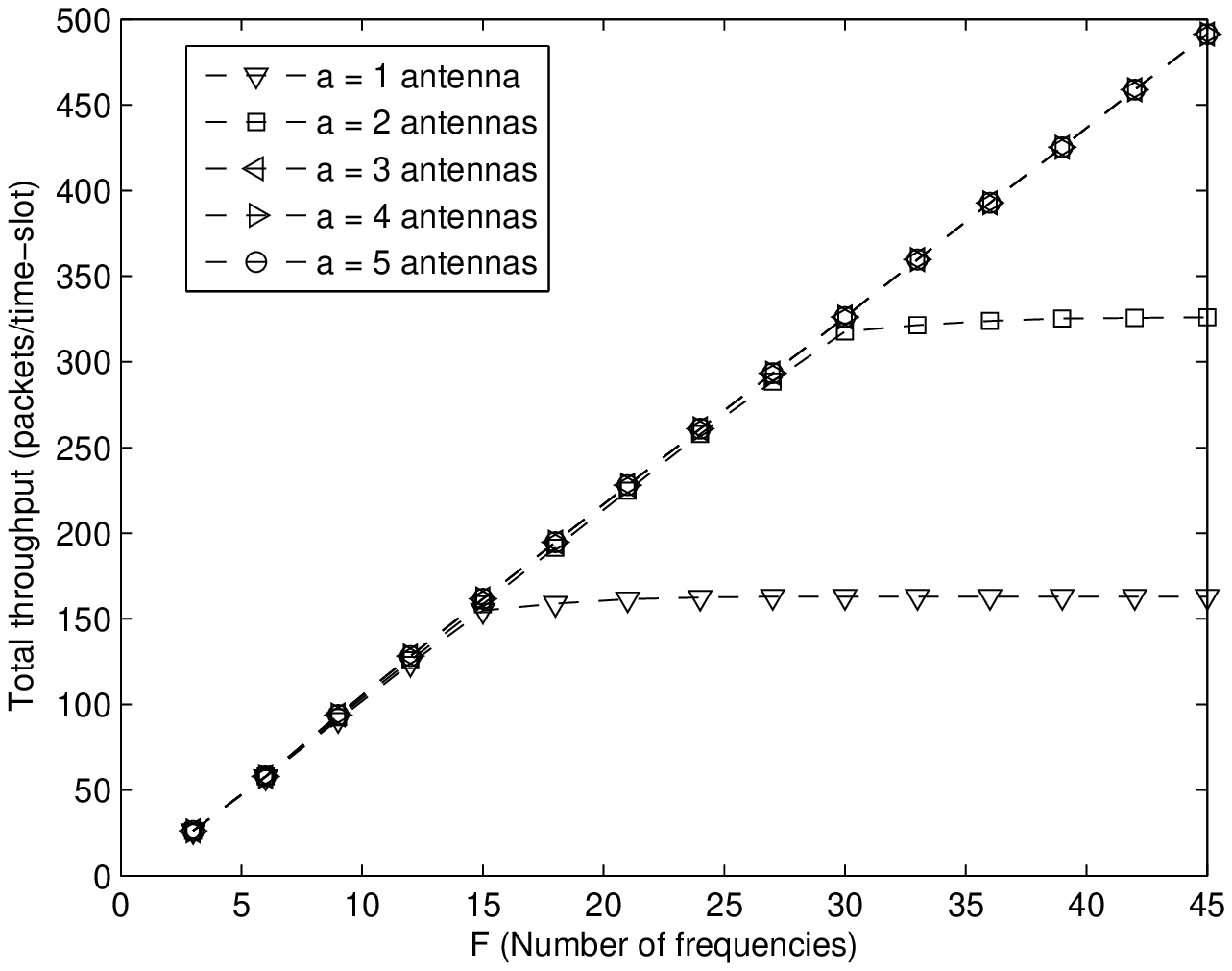}
\caption{Average total network throughput for throughput maximizing scheduler for varying $F$ and $a$.} \label{fig:factorFa-2D-EN}
\end{center}
\end{figure}
\FloatBarrier

\begin{figure*}[!htb]
\centering
\subfigure[Average throughput of all the SUs.]{
\includegraphics[scale=0.5]{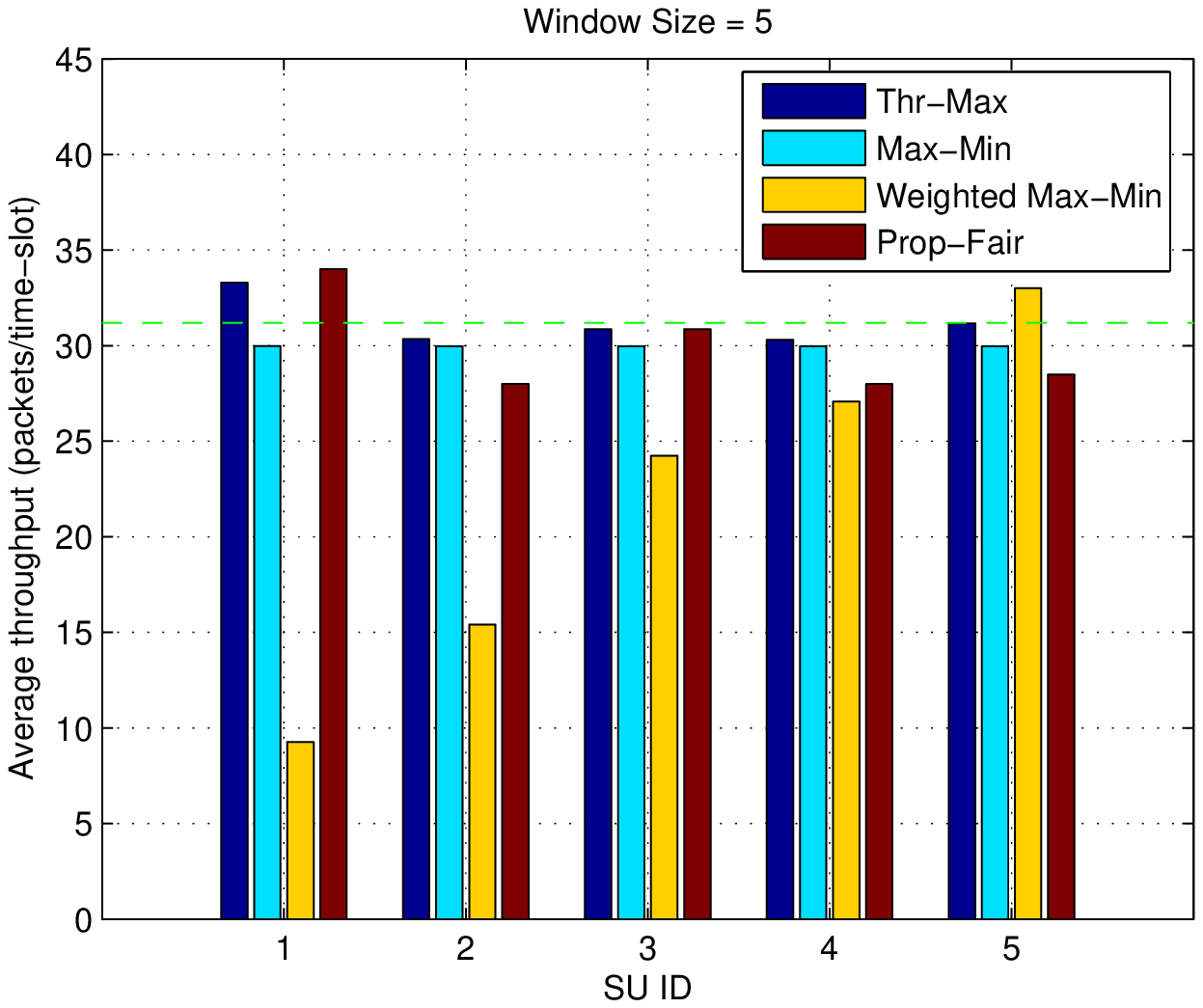}
\label{fig:FairSch_5SUs_WindowSize5}
}
\subfigure[Average total throughput.]{
\includegraphics[scale=0.5]{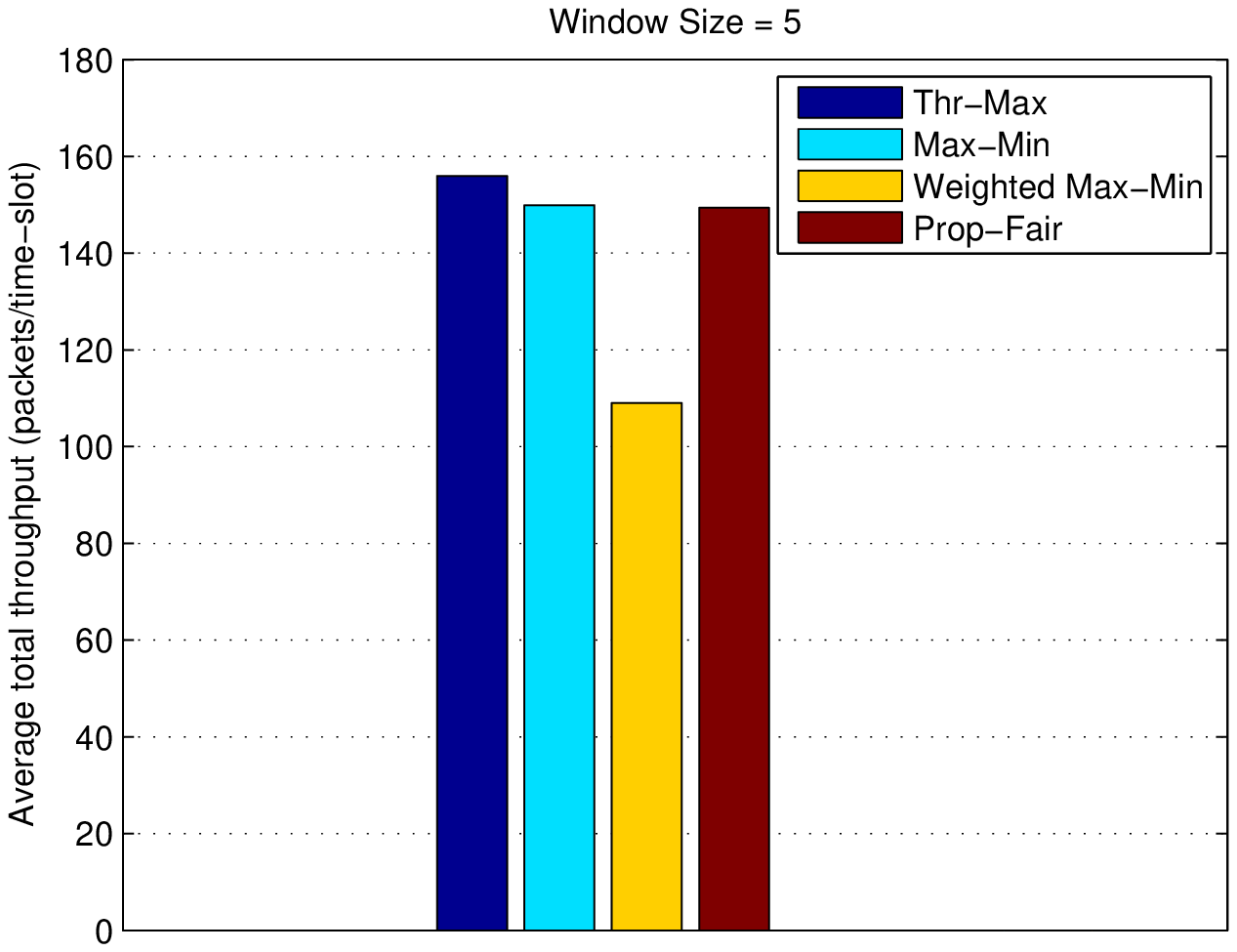}
\label{fig:FairSch_5SUs_WindowSize5_TotalThroughput}
}
\caption[Optional caption for list of figures]{All schedulers for $N = 5$, window size ($\varphi$) = $5$, and target weights: $\eta_{1} = 0.05$, $\eta_{2} = 0.1$, $\eta_{3} = 0.2$,
$\eta_{4} = 0.25$, $\eta_{5} = 0.4$. }
\label{fig:FairSch_BarGraph}
\end{figure*}
\FloatBarrier

\begin{table*}[!htb]
  \centering
\footnotesize\addtolength{\tabcolsep}{-1pt}
\begin{center}
  \begin{tabular}{ | c | c | c | c | c | c | c | c | c | c | c | c | c | }
    \hline
    SU index & $\varphi = 1$ & $\varphi = 5$ & $\varphi = 10$ & $\varphi = 15$ & $\varphi = 20$ & $\varphi = 25$ & $\varphi = 30$ & $\varphi = 35$ & $\varphi = 40$ & $\varphi = 45$ & $\varphi = 50$ & Target \\
 & & & & & & & & & & & & Weight ($\eta$)\\ \hline
   SU-1 & 0.057 & 0.085 & 0.108 & 0.121 & 0.133 & 0.142 & 0.148 & 0.153 & 0.157 & 0.16 & 0.163 & 0.05\\ \hline
   SU-2 & 0.105 & 0.141 & 0.162 & 0.174 & 0.179 & 0.184 & 0.186 & 0.188 & 0.189 & 0.191 & 0.192 & 0.1\\ \hline
   SU-3 & 0.195 & 0.222 & 0.221 & 0.219 & 0.217 & 0.215 & 0.213 & 0.212 & 0.211 & 0.21 & 0.21 & 0.2\\ \hline
   SU-4 & 0.249 & 0.248 & 0.238 & 0.231 & 0.226 & 0.222 & 0.22 & 0.218 & 0.216 & 0.215 & 0.213 & 0.25\\ \hline
   SU-5 & 0.391 & 0.302 & 0.268 & 0.252 & 0.242 & 0.235 & 0.23 & 0.227 & 0.224 & 0.222 & 0.22 & 0.4\\ \hline
  \end{tabular}
\end{center}
  \caption{Achieved and target throughput ratios for weighted max-min fair scheduler}\label{table:weightedMMFS}
\end{table*}
\FloatBarrier

\begin{figure*}[!htbp]
\centering
\subfigure[Average total throughput.]{
\includegraphics[scale=0.5]{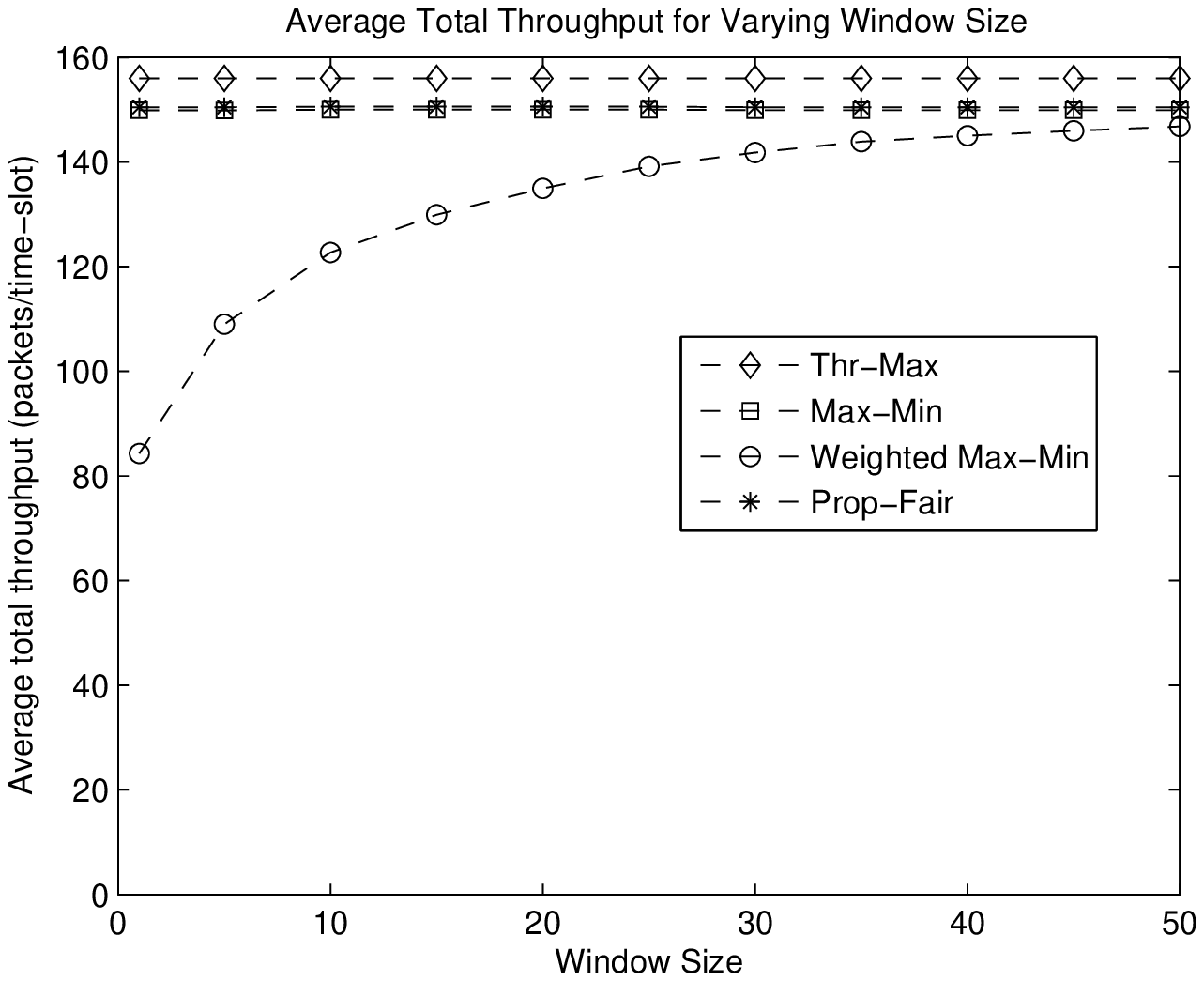}
\label{fig:FairSch_5SUs_VaryingWindowSize_TotThr}
}
\subfigure[Average Jain fairness index.]{
\includegraphics[scale=0.5]{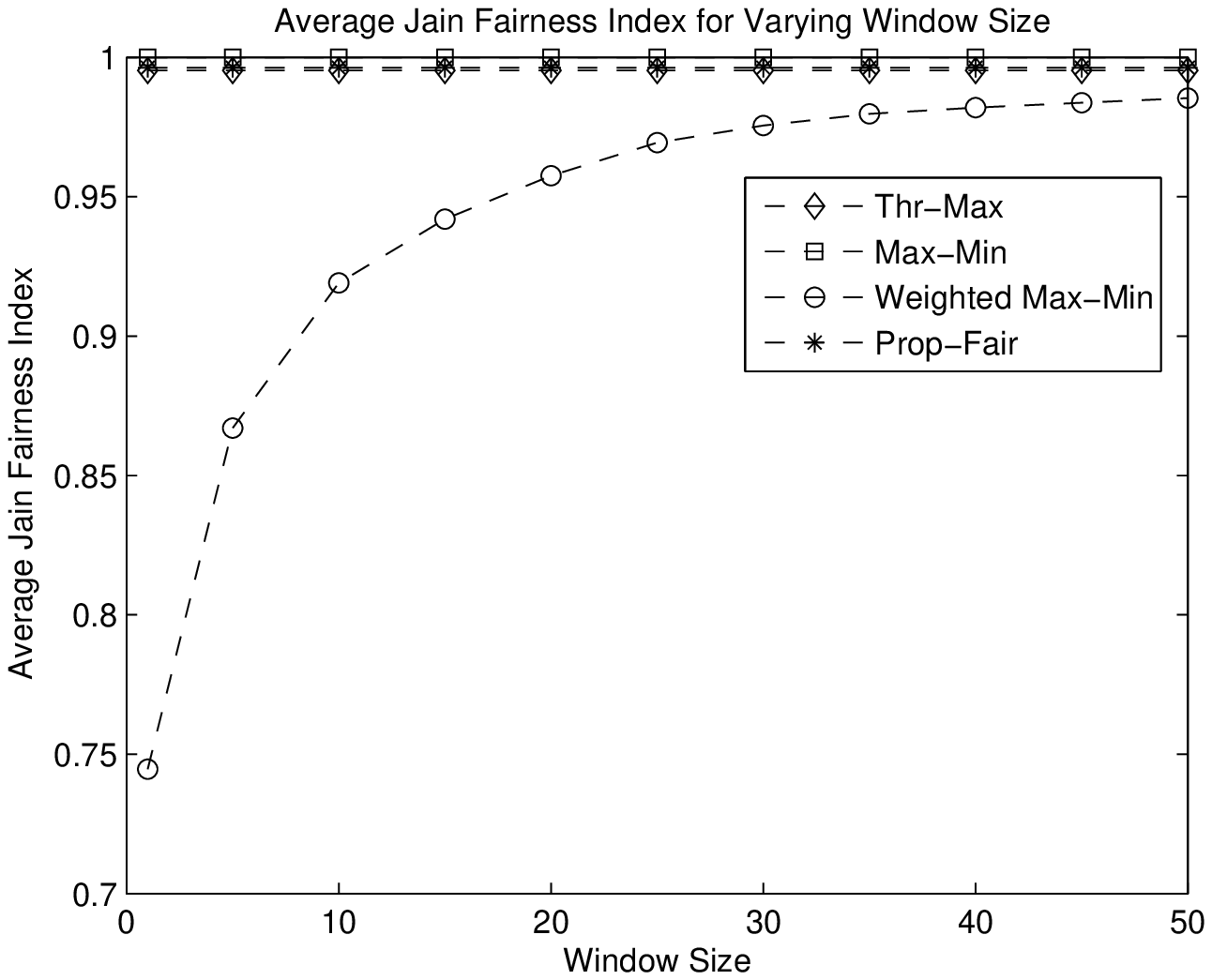}
\label{fig:FairSch_5SUs_WindowSize5_Jain}
}
\caption[Optional caption for list of figures]{All schedulers for varying window size ($\varphi$), $N=5$, and target weights: $\eta_{1} = 0.05$, $\eta_{2} = 0.1$, $\eta_{3} = 0.2$, $\eta_{4} = 0.25$, $\eta_{5} = 0.4$. }
\label{fig:FairSch_VaryingWindowSize}
\end{figure*}
\FloatBarrier

\begin{figure*}[!htbp]
\centering
\subfigure[Average total throughput.]{
\includegraphics[scale=0.5]{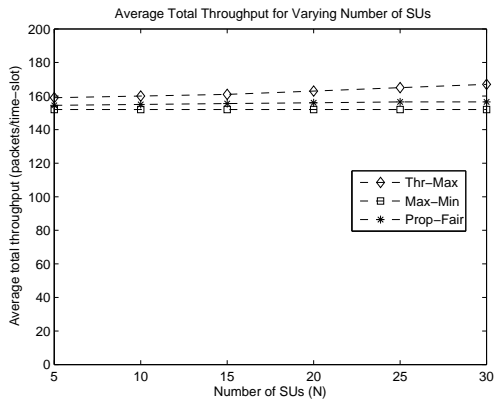}
\label{fig:FairSch_VaryingNumSUs_TotThr}
}
\subfigure[Average Jain fairness index.]{
\includegraphics[scale=0.5]{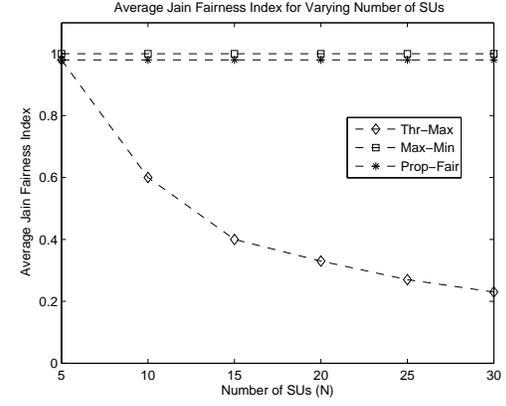}
\label{fig:FairSch_VaryingNumSUs_Jain}
}
\caption[Optional caption for list of figures]{Throughput maximizing, max-min fair, and proportionally fair schedulers for varying number of SUs, $F=15$, and $\varphi = 5$. }
\label{fig:FairSch_VaryingNumSUs}
\end{figure*}
\FloatBarrier

\begin{figure*}[!htbp]
\centering
\subfigure[Performance of max-min fair scheduling]{
\includegraphics[scale=0.5]{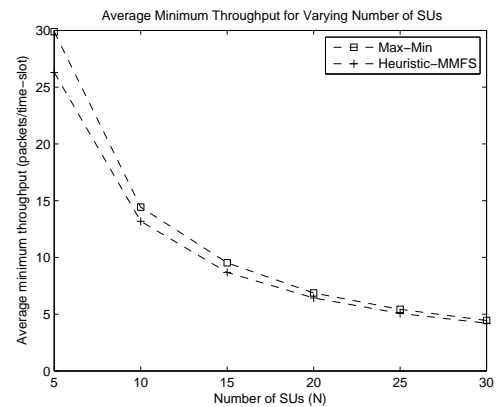}
\label{fig:HeuristicMMFS_VaryingNumSUs}
}
\subfigure[Performance of proportionally fair scheduling]{
\includegraphics[scale=0.5]{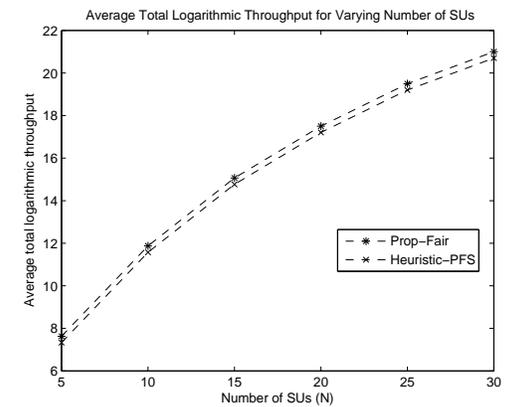}
\label{fig:HeuristicPFS_VaryingNumSUs}
}
\caption[Optional caption for list of figures]{Performance of our heuristic algorithm for varying number of SUs}
\label{fig:HeuristicVaryingNumSUs}
\end{figure*}
\FloatBarrier

\begin{figure*}[!htbp]
\centering
\subfigure[Performance of max-min fair scheduling]{
\includegraphics[scale=0.5]{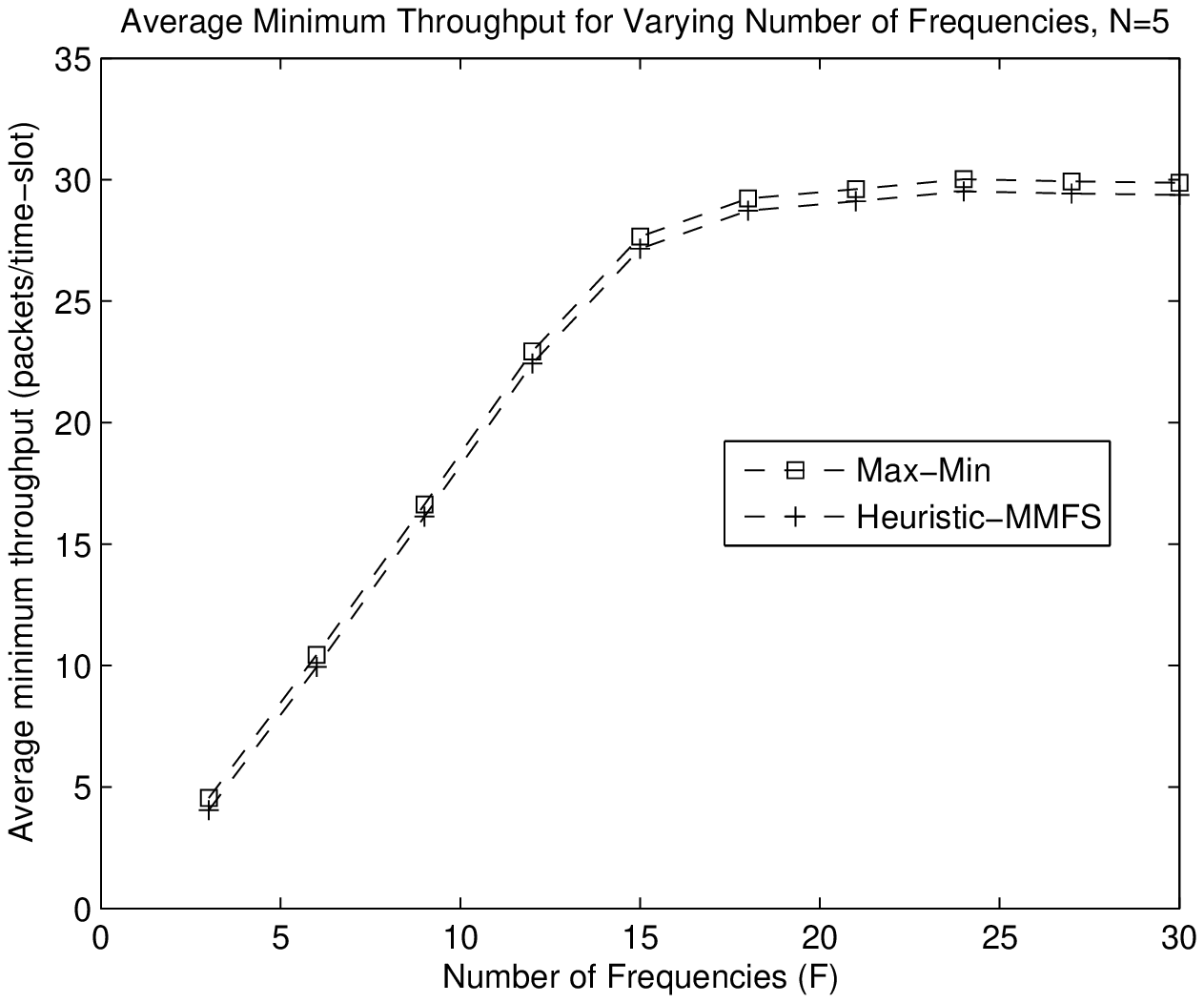}
\label{fig:HeuristicMMFS_VaryingNumFreqsN5}
}
\subfigure[Performance of weighted max-min fair scheduling]{
\includegraphics[scale=0.5]{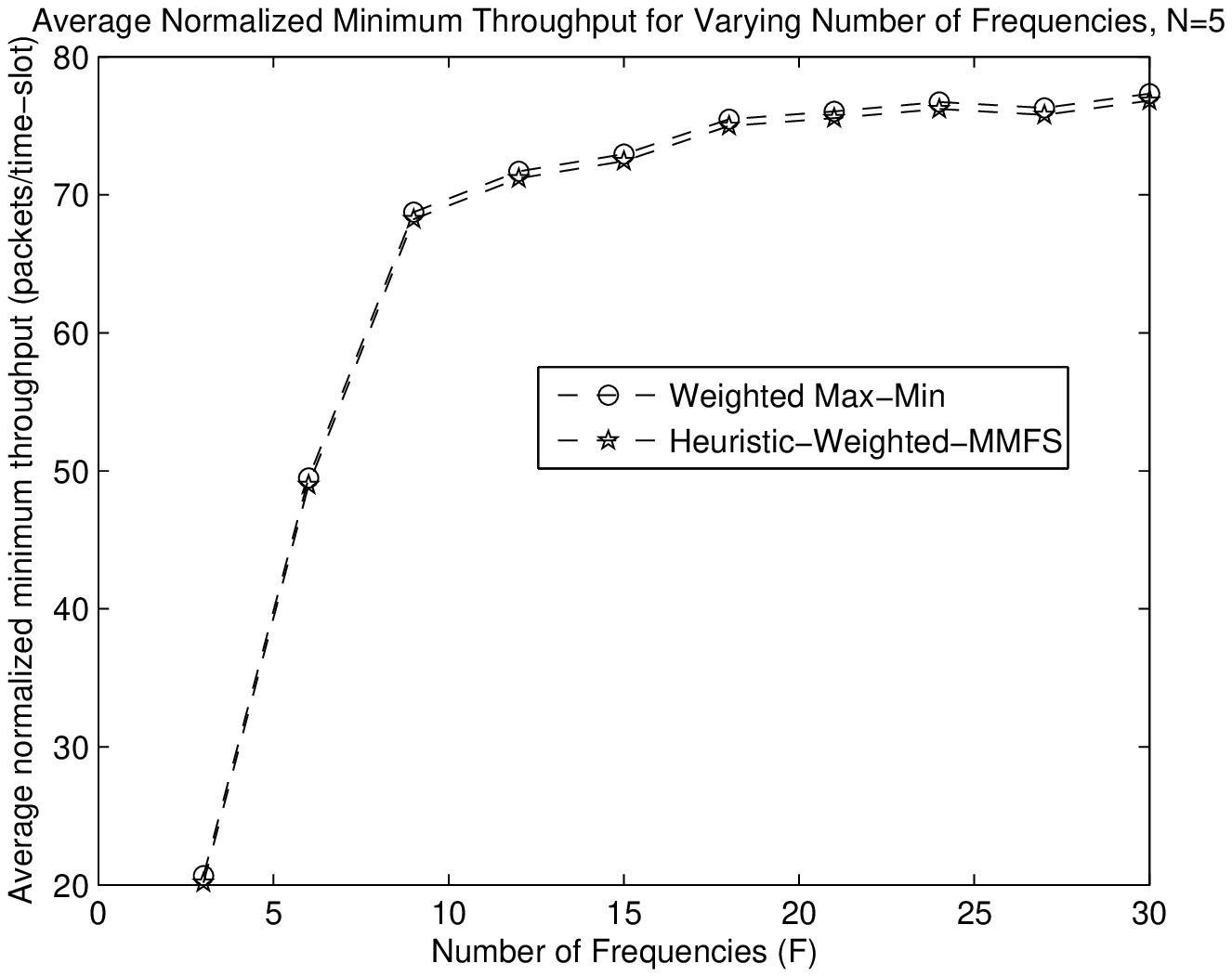}
\label{fig:HeuristicWeightedMMFS_VaryingNumFreqsN5}
}
\subfigure[Performance of proportionally fair scheduling]{
\includegraphics[scale=0.5]{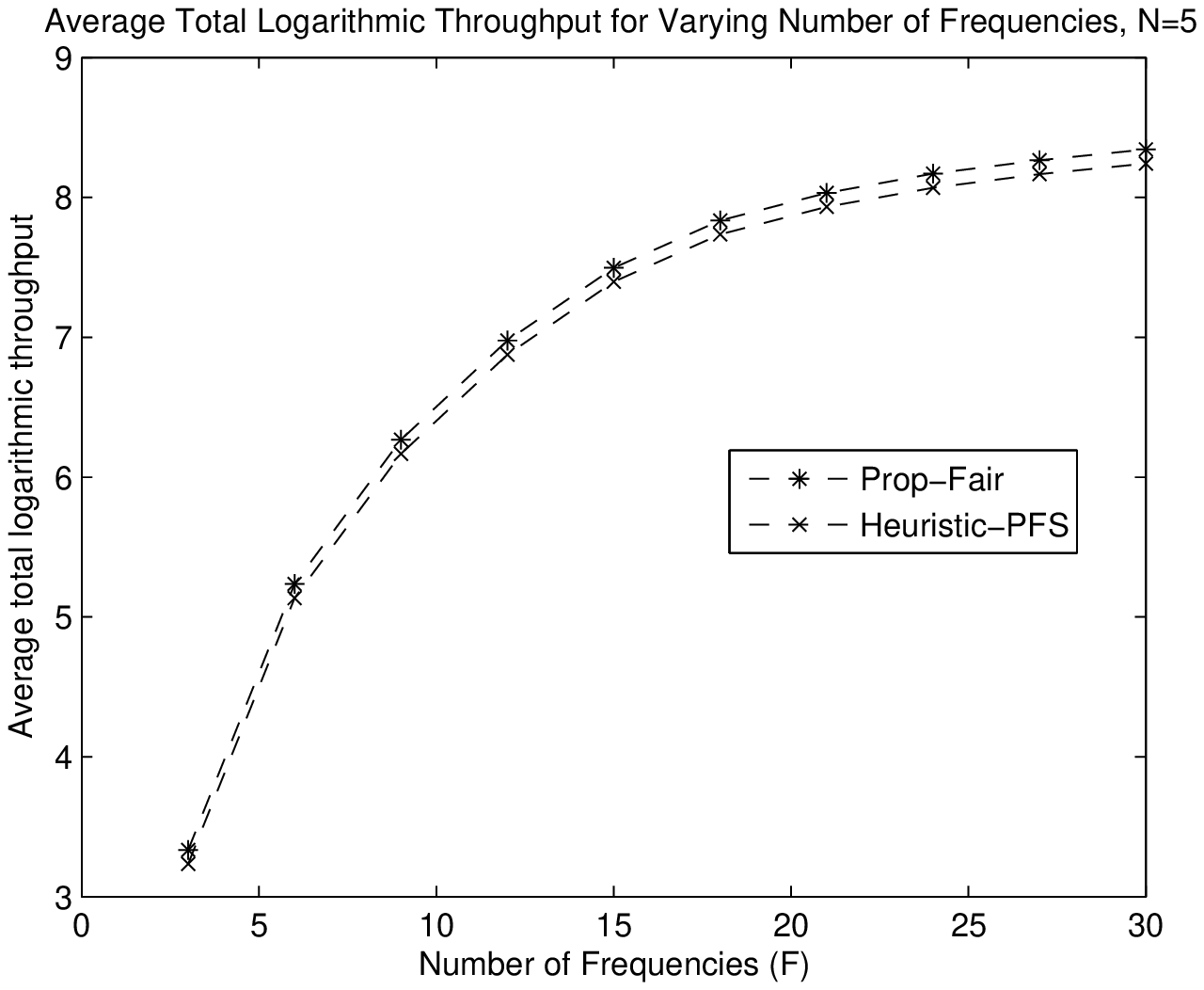}
\label{fig:HeuristicPFS_VaryingNumFreqsN5}
}
\caption[Optional caption for list of figures]{Performance of our heuristic algorithm for N=5 and varying number of frequencies}
\label{fig:HeuristicVaryingNumFreqsN5}
\end{figure*}
\FloatBarrier

\begin{figure*}[h]
\centering
\subfigure[Performance of max-min fair scheduling]{
\includegraphics[scale=0.5]{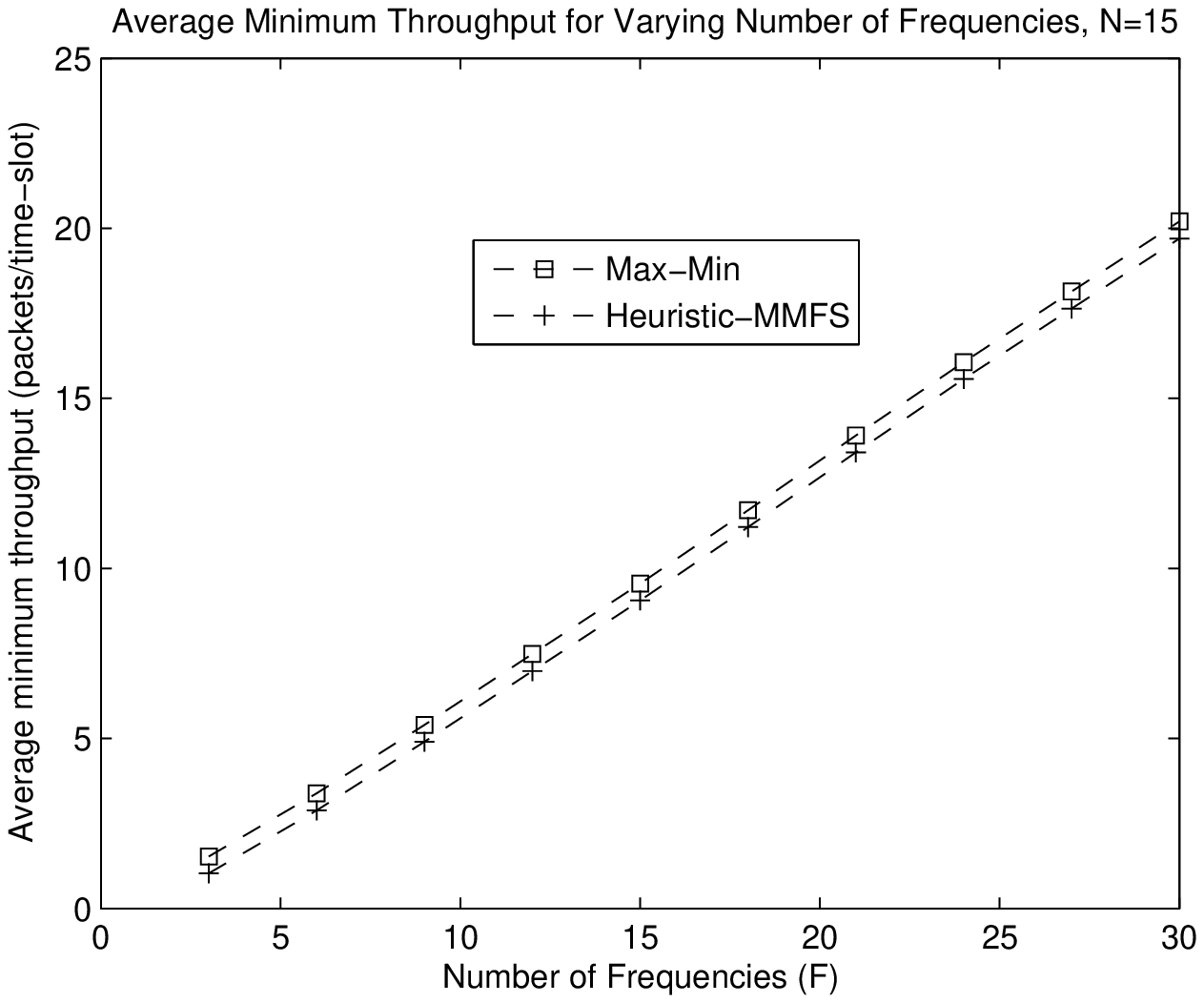}
\label{fig:HeuristicMMFS_VaryingNumFreqsN15}
}
\subfigure[Performance of proportionally fair scheduling]{
\includegraphics[scale=0.5]{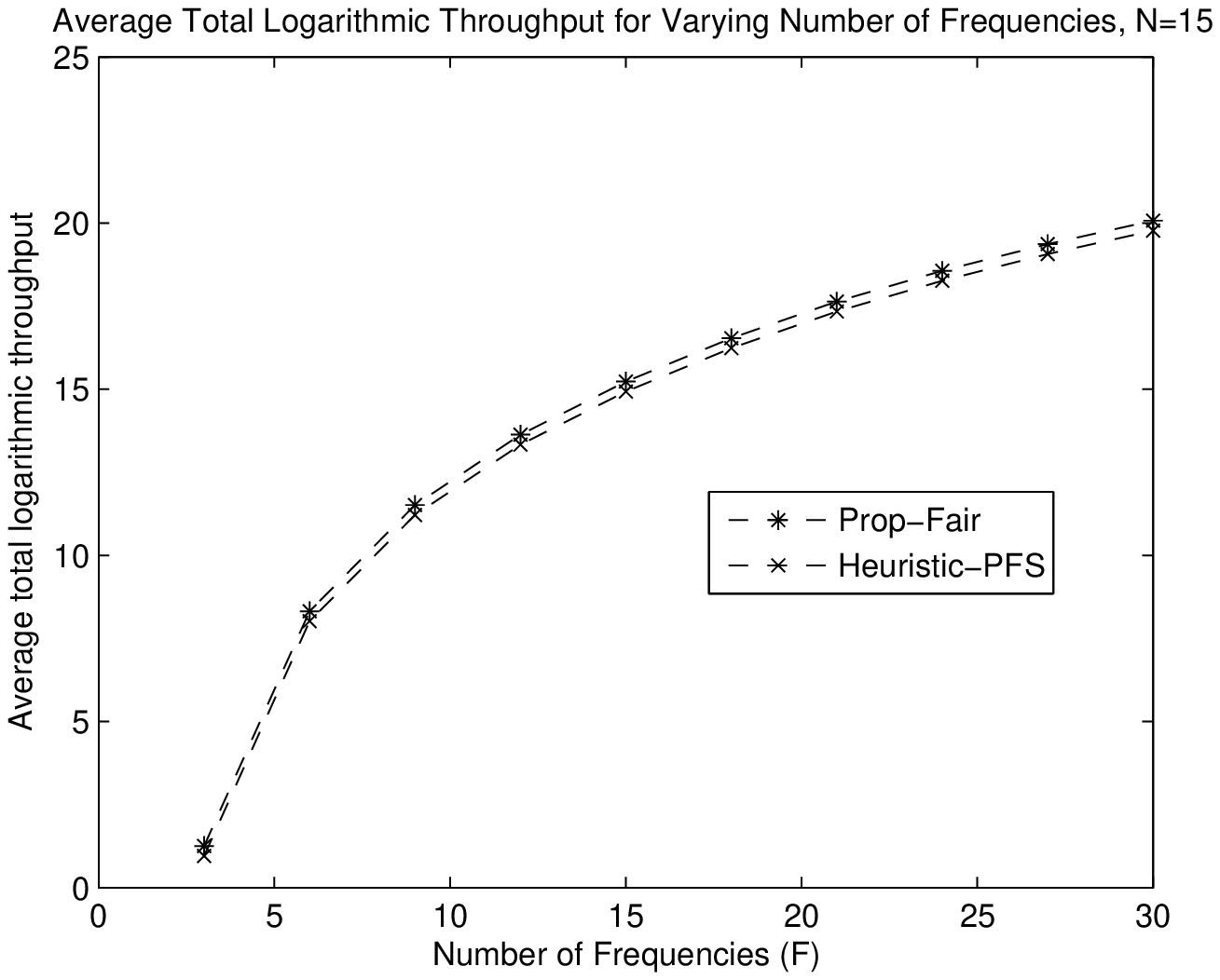}
\label{fig:HeuristicPFS_VaryingNumFreqsN15}
}
\caption[Optional caption for list of figures]{Performance of our heuristic algorithm for N=15 and varying number of frequencies}
\label{fig:HeuristicVaryingNumFreqsN15}
\end{figure*}
\FloatBarrier

\end{document}